\documentclass[journal]{IEEEtran}

\usepackage{amssymb}
\usepackage{amsthm}
\usepackage{amsmath}
\usepackage{xcolor}
\usepackage{comment}
\usepackage{nomencl}
\usepackage{ifthen}
\usepackage{etoolbox}
\usepackage{bbm}
\usepackage{algorithm}
\usepackage{algpseudocode}
\usepackage{enumerate}
\usepackage{float}
\usepackage{booktabs}
\usepackage{tikz}
\usepackage{subcaption}
\usepackage{soul}
\usepackage{multirow}
\usepackage{algpseudocode}
\usepackage{url}
\usepackage{cite}
\DeclareMathSymbol{\Minus}{\mathbin}{AMSa}{"39}
\def\Equal{\texttt{=}}
\def\Plus{\texttt{+}}
\newcommand{\In}{\mbox{$\in$}}
\theoremstyle{remark} 
\newtheorem{remark}{Remark}
\allowdisplaybreaks

\usetikzlibrary{shapes.geometric}
\usetikzlibrary{arrows}
\usetikzlibrary{calc}

\tikzstyle{startstop} = [rectangle, rounded corners, 
minimum width=1cm, 
minimum height=0.5cm,
text centered, 
draw=white, 
fill=white!30]

\tikzstyle{process} = [rectangle, 
minimum width=1cm, 
minimum height=0.5cm, 
text centered, 
text width=2cm, 
draw=black, 
fill=orange!30]

\tikzstyle{arrow} = [thick,->,>=stealth]

\tikzset{sin v source/.style={circle,draw,append after command={
    \pgfextra{
    \draw
      ($(\tikzlastnode.center)!0.5!(\tikzlastnode.west)$)
       arc[start angle=180,end angle=0,radius=0.425ex] 
      (\tikzlastnode.center)
       arc[start angle=180,end angle=360,radius=0.425ex]
      ($(\tikzlastnode.center)!0.5!(\tikzlastnode.east)$) 
    ;}},scale=2,}}

\tikzset{w source/.style={circle,draw,scale=1}}

\makenomenclature
\renewcommand{\nomgroup}[1]{%
  \item[%
    \ifthenelse{\equal{#1}{A}}{A. \textit{Sets}}{}%
    \ifthenelse{\equal{#1}{B}}{B. \textit{Parameters}}{}%
    \ifthenelse{\equal{#1}{C}}{C. \textit{Variables}}{}
    ]%
    \vspace{10pt}\hspace*{-\leftmargin}\vspace{10pt}%
}

\newtheorem{theorem}{Theorem}
\newtheorem{proposition}{Proposition}
\theoremstyle{plain}

\begin{document}
%
\title{Electricity Market-Clearing With Extreme Events} 
%
\author{Tom\'as~Tapia, 
    Zhirui~Liang,
    Charalambos~Konstantinou,Yury~Dvorkin
}

\maketitle

\begin{abstract}
Extreme events jeopardize power network operations, causing beyond-design failures and massive supply interruptions. Existing market designs fail to internalize and systematically assess the risk of extreme and rare events. Efficiently maintaining the reliability of renewable-dominant power systems during extreme weather events requires co-optimizing system resources, while differentiating between large/rare and small/frequent deviations from forecast conditions. To address this gap in both research and practice, we propose managing the uncertainties associated with extreme weather events through an additional reserve service, termed extreme reserve. The procurement of extreme reserve is co-optimized with energy and regular reserve using a large deviation theory chance-constrained (LDT-CC) model, where LDT offers a mathematical framework to quantify the increased uncertainty during extreme events. To mitigate the high additional costs associated with reserve scheduling under the LDT-CC model, we also propose an LDT model based on weighted chance constraints (LDT-WCC). This model prepares the power system for extreme events at a lower cost, making it a less conservative alternative to the LDT-CC model. The proposed market design leads to a competitive equilibrium while ensuring cost recovery. Numerical experiments on an illustrative system and a modified 8-zone ISO New England system highlight the advantages of the proposed market design.

\end{abstract}


\IEEEpeerreviewmaketitle


\section{Introduction}\label{Sec:Introduction}
\subsection{Motivation and Scope}
\IEEEPARstart{R}ARE and extreme events are situations that occur with a low probability but can lead to catastrophic system impacts, provoking cascading blackouts and affecting both the economy and society \cite{PANTELI2015259}. For example, the Federal Energy Regulatory Commission reports that during the extreme cold winter storm in Texas in February 2021, cold temperatures severely impacted power generation capacity, leading to energy shortages and causing damages estimated between 80 to 130 billion dollars \cite{ferc2021february}. Also, during the three-day storm, wholesale electricity prices often surged to the offer price cap of $\$$9,000/MWh \cite{LEVIN20221}. After this episode, the Public Utility Commission of Texas lowered the offer price cap from $\$$9,000/MWh to $\$$5,000/MWh and imposed strict weatherization standards on generation and natural gas companies \cite{TexasSB3}, \cite{PUCT_2021}. While seemingly beneficial to consumers, the offer price cap measure may cause market participants and investors to perceive electricity prices as lower than they would be in a fully risk-complete market. This misalignment between private (investor) and social (system) risk attitudes may exacerbate the missing money problem. In the long term, this measure may also negatively impact resource adequacy, leading to insufficient capacity investment to ensure system reliability during future extreme events \cite{mays2022private}.

Other weather events such as wildfires, heatwaves and hurricanes, dunkelflaute phenomena, and cyber-attacks are also classified as extreme weather events and are shown to affect electricity market outcomes \cite{avraam2023operational}.These extreme events affect various components of the power grid, including power generation capacity \cite{ke2016quantifying}, transmission capacity \cite{karimi2018dynamic}, and energy storage capacity \cite{ma2018temperature}. However, a common characteristic of these events is their unpredictability, especially for power system look-ahead scheduling. We refer to an \textit{extreme event} as a singular instance of uncertainty in the power system, defined by its magnitude, location, and duration, with the potential to cause significant disruptions to system operation \cite{mujjuni2023evaluation}. 

Current electricity markets set reserve requirements exogenously and then enforce them in scheduling routines without an explicit treatment of extreme events. As a result, these (often heuristic) reserve rules do not cover extreme events or prioritize resources for rare or large deviations from forecast or design conditions, resulting in risk-incomplete market outcomes. This incompleteness, in turns, leads to widespread outages and costly operating regimes and inadequate dispatch and price signals to market participants. Accounting for extreme events within market mechanisms is crucial given the expected increase in the frequency, intensity, and duration of these events due to climate change \cite{PANTELI2015259}.

Recent studies \cite{mays2023financial,mays2022private,billimoria2023contract,bienstock2024risk} have advanced the understanding of financial risk management and resilience under extreme conditions, emphasizing the need to adjust pricing and market mechanisms to address revenue volatility and the costs of managing extreme risks. Specifically, \cite{mays2023financial} points out that decentralized markets often face under-investment in resilience due to market frictions. The authors in \cite{mays2022private} and \cite{billimoria2023contract} noted the importance of regulatory frameworks and financial tools, such as forward contracts and risk-sharing instruments, to promote investment in resources needed to cope with extreme events without distorting price signals or discouraging competition. Additionally, \cite{bienstock2024risk} emphasizes the need to reduce  customer exposure to price volatility which is often driven by extreme weather patterns.

This paper addresses the system's lack of preparedness for rare and extreme events by introducing a new reserve service, which we refer to as \textit{extreme reserve}. In contrast, we refer to margins (e.g., load following or regulation) scheduled in traditional look-ahead markets (e.g., day-ahead, intra-day or hour-ahead) as \textit{regular reserve} \cite{khatami2019flexibility}. To this end, we develop a chance-constrained (CC) system scheduling model and use the large deviation theory (LDT) to capture the significant uncertainty posed by extreme events caused by weather-dependent renewable generation resources. The resulting  LDT-CC model enables the scheduling of both extreme and regular reserves while deriving the marginal prices for energy and reserve services. These prices lead to a competitive market equilibrium.

\subsection{Literature Review}
Over the past few decades, optimization techniques for managing uncertainty in power systems scheduling and market clearing have evolved rapidly \cite{roald2023power}, including stochastic programming \cite{papavasiliou2011reserve}, robust optimization techniques \cite{lorca2014adaptive}, chance (probabilistic) constraints, and distributionally robust optimization \cite{bienstock2014chance}. Still, the current industry practice remains largely deterministic and aims to cope with growing uncertainty through incremental improvements, thereby increasing complexity and opaqueness of operating procedures and software \cite{hobbs2019three}. In contrast, stochastic electricity market designs make it possible to internalize uncertainty and provide efficient market signals. 

These market designs achieve competitive equilibrium under various uncertainty factors and assumptions, with market signals—primarily derived from prices—playing a critical role in market clearing mechanisms to align private and social risk perspectives \cite{mays2022private}. However, scenario-based stochastic programming faces significant limitations for market clearing routines due to scenario dependency and computational barriers \cite{papavasiliou2014applying}. It requires the use of \textit{nontransparent} scenario selection techniques and scenario weighting to avoid biasing the results \cite{papavasiliou2013multiarea}, as well as the inability to accurately predict scenarios for extreme events \cite{tong2022optimization}. Alternatively, robust and distributionally robust optimization can capture extreme events but typically lead to \textit{overly conservative} solutions, resulting in suboptimal asset- and system-level operations. 

Chance constraints (CCs) are a reliable method to manage and price resources effectively, addressing risks by employing (often) affine control policies to determine the necessary reserve capacity in response to a priori postulated uncertainty \cite{bienstock2014chance}. CCs also  position the system to cope with anticipated uncertainty realizations by limiting constraint violations to only a small fraction of the time \cite{bienstock2014chance}. This method has been extended further to robust CCs \cite{lubin2015robust}, distributionally robust CCs \cite{xie2017distributionally}, and used for endogenous electricity pricing \cite{dvorkin2019chance,kuang2018pricing,fang2019introducing,mieth2020risk,wernerpricing,liang2022inertia}. Despite their strengths, CCs are generally indifferent to the explicit risk associated with the impact or size of constraint violations, particularly overlooking the risk of \textit{large} or \textit{rare} deviations.  This can lead to risk-incomplete solutions that are both costly and ineffective in managing extreme events.

Notably, \cite{roald2015optimal} introduces weighted chance constraints (WCC) with general (non-affine) and, importantly for pricing, convex control policies that differentiate the response of generators between large and small deviations due to uncertainty. Related to \cite{roald2015optimal}, \cite{porras2023integrating} presents a sample-based model for calculating additional manual reserves. However, these approaches yield an NP-hard problem and require approximations to be computed efficiently. Such  approximations have been studied in recent literature, e.g, \cite{jiang2022also,hanasusanto2017ambiguous,nemirovski2007convex,ahmed2018relaxations}. Still, these approximations complicate the solving process due to the use of non-convex or scenario-based methods, particularly when trying to account for rare events. 

Typical events occur with a relatively high probability and are described by the Law of Large Numbers or the Central Limit Theorem, which explain how averages of random variables converge to their expected value. Extreme or rare events, however, deviate significantly from this expected value, and Large Deviation Theory (LDT) offers tools to estimate the probability of such deviations \cite{dematteis2019extreme}. 
The central element in LDT is the rate function, denoted by $I(x)$, where $x$ is a possible outcome of a random process. The rate function governs the exponential decay of the probability of rare events, capturing how the likelihood of different deviations from typical behavior decreases. It is a non-negative, convex function that reaches its minimum at the most likely outcomes.

We first interpret the rate function in mathematical terms. 
Consider a sequence of random variables $\{X_n\}$ in a stochastic process or time series, where $n$ represent time step. According to  LDT, the probability that $X_n$ takes values in a rare event set $\mathcal{A}$ decays exponentially as $n$ increases. More formally, LDT seeks to approximate $\mathbb{P}(X_n \in \mathcal{A})$ by
$\exp\left(-n \inf_{x \in \mathcal{A}} I(x)\right)$ as $n \to \infty$. The quantity $\inf_{x \in \mathcal{A}} I(x)$ represents the ``most probable'' way for the system to deviate into  set $\mathcal{A}$. If $\mathcal{A}$ contains  point $x^* =\arg\min I(x)$, then $x^*$ is the most likely large deviation, referred to as the \textit{dominating point}. This dominating point is essential in estimating the probability of rare events.

We then provide an intuitive explanation of the dominating point in the context of power systems. Consider a thermal generator with upper and lower limits on its output. Under normal conditions, this generator operates within a safe range, but fluctuations in renewable power generation may cause it to briefly operate outside this range during real-time dispatch. Among all the possible ways the generator can deviate from its safe operating region, some deviations are more likely than others. The dominating point refers to the most probable way that deviation occurs.

The notion of the dominating point is useful for quantifying reserve requirements for rare and extreme events in power systems. Although multiple events could push the system out of its safe operating region, we only need to focus on the most critical scenario — the one with the highest probability of driving the system into the critical operating region. By scheduling extreme reserve for this critical scenario, and controlling the probability of constraint violations, we obtain a Large Deviation Theory-based Chance Constraint (LDT-CC).

The LDT-CC proposed in \cite{tong2022optimization} introduced a \textit{sample-free} approach for quantifying and incorporating rare and extreme events, addressing the computational challenges posed by traditional methods that rely on intensive sampling. This approach leads to a bi-level optimization formulation that is independent of the rarity of the event. To handle the bi-level structure, the lower-level problem is replaced with its first-order optimality conditions, resulting in a convex model that can be efficiently solved by off-the-shelf optimization solvers.

Fig.~\ref{fig:IntroProbCC} compares the risk hedging strategies of a regular CC and an LDT-CC under uncertainty $\boldsymbol{\Omega}$. The red area represents deviations occurring during regular scenarios, where the cumulative probability ($1-\epsilon$) is hedged by the regular CC. In contrast, the yellow area corresponds to extreme scenarios, where the LDT-CC robustly hedges the uncertainty of extreme events by only using a dominating point ($\Omega^*$) to characterize the rare event set in the tail of the distribution ($\epsilon$). The use of $\Omega^*$ simplifies the analysis of extreme events by reducing the problem to a single point that captures the system's essential behavior, minimizing complexity and computational burden \cite{tong2021extreme}. For comparison, we include the Value-at-Risk (VaR) and Conditional Value-at-Risk (CVaR) metrics \cite{nemirovski2007convex} in the figure.
\begin{figure}[!t]
    \centering\includegraphics[width=0.3\textwidth]{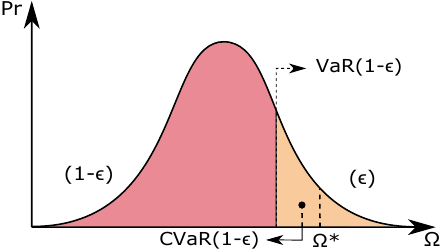}
    \caption{Operation regimes under uncertainty with regular (red) and LDT (yellow) chance constraints.}
    \label{fig:IntroProbCC}
\end{figure}

\subsection{Contributions}
Our contributions in this paper include:
\begin{itemize}
\item We extend the chance-constrained pricing approaches from \cite{dvorkin2019chance, mieth2020risk, kuang2018pricing} by incorporating LDT-CC into an economic dispatch (ED) model. The LDT-CC-ED model effectively manages uncertainty arising from rare and extreme events in the market-clearing process. Its convexity with respect to power generation variables enables efficient solving using standard optimization solvers.
\item We aim to relax the extreme reserve scheduling problem by imposing less stringent risk hedging requirements during extreme events. To this end, we propose a LDT weighted-chance constrained ED model (LDT-WCC-ED) to achieve a less conservative extreme reserve scheduling, reducing operational costs while maintaining acceptable reliability under extreme conditions.
\item We achieve market clearing using the proposed LDT-CC-ED and LDT-WCC-ED models, deriving the marginal prices for energy, regular reserve, and extreme reserve. We also demonstrate that the resulting market clearing is efficient and establishes a competitive equilibrium.
\end{itemize}


\nomenclature[A]{$\mathcal{N}$}{Set of controllable generating units, indexed by $n$.}%
\nomenclature[A]{$\mathcal{W}$}{Set of wind generators, indexed by $n'$.}%
\nomenclature[A]{$\mathcal{W}_i$}{Set of wind generators at node $i$.}%
\nomenclature[A]{$\mathcal{I}$}{Set of buses, indexed by $i$.}%
\nomenclature[A]{$\mathcal{L}$}{Set of lines, indexed by $(j,k)$ that represent the nodes connected.}%
\nomenclature[A]{$\mathcal{I'}$}{Set of buses with wind generators, indexed by $i$.}%
\nomenclature[A]{$\mathcal{N}_i$}{Set of controllable generators at node $i$.}%
\nomenclature[A]{$\mathcal{L}_i^{+}$}{Set of lines with pair of nodes $(i,j)$.}%
\nomenclature[A]{$\mathcal{L}_i^{-}$}{Set of lines with pair of nodes $(j,i)$.}%
\nomenclature[B]{$p_{n}^{\max}$}{Maximum output limit of unit $n$ [MW].}%
\nomenclature[B]{$p_{n}^{\min}$}{Minimum output limit of unit $n$ [MW].}%
\nomenclature[B]{$f_{jk}^{\max}$}{Maximum flow limit of line between the nodes $j$ and $k$ [MW].}%
\nomenclature[B]{$d_i$}{Demand at node $i$ [MW].}%
\nomenclature[B]{$D$}{Aggregated demand [MW].}%
\nomenclature[B]{$\epsilon_n$}{Probability of the output of the unit $n$
exceeding the max/min limit [$\%$].}%
\nomenclature[B]{$\hat{W}$}{Aggregated wind forecast generation [MW]}%
\nomenclature[B]{$\hat{W}_i$}{Aggregated wind forecast generation at node $i$ [MW]}%

\nomenclature[C]{$p_{n}$}{Power output dispatch of unit $n$ [MW].}%
\nomenclature[C]{$\alpha_{n}$}{Participation factor of unit $n$ for the total wind deviation $\Omega$ [$\%$].}%
\nomenclature[C]{$\beta_{n}$}{Participation factor of unit $n$ for the large wind deviation $\Omega^*$ [$\%$].}%
\nomenclature[C]{$\Omega_i$}{Aggregated wind forecast error at node $i$ [MW].}%
\nomenclature[C]{$\Omega$}{Aggregated wind forecast error [MW].}%
\nomenclature[C]{$\Omega^*$}{LDT minimizer for the wind uncertainty realization $\Omega$ [MW].}%
\nomenclature[C]{$\Omega^*_i$}{LDT minimizer for the wind uncertainty realization $\Omega_i$ at node $i$ [MW].}%
\nomenclature[C]{$\lambda^*_n$}{Optimal dual variable of the lower-level LDT problem for the unit $n$.}%
\nomenclature[C]{$f_{jk}$}{Power flow in the line that connect the nodes $jk$ [MW].}%
\nomenclature[C]{$\theta_i$}{DC simplification angle at node $i$.}%
\nomenclature[C]{$\delta_{n}(\Omega)$}{general control law of generator $n$ for the total deviation of wind farms $\Omega$.}%
\nomenclature[C]{$g_{n}(\Omega)$}{general control law of generator $n$ for the total deviation of wind farms $\Omega$.}%

\section{Benchmark Models} \label{Sec:Benchmark_Models}
This section first reviews the chance-constrained economic dispatch (CC-ED) model in Section~\ref{Subsection:CC}. This model ensures power balance in the system while limiting the rate of constraint violations under uncertainty. Flexible resources provide regular reserve to mitigate the risk of uncertainty, following an affine control policy. In Section~\ref{Section:WCC}, we extend CCs to weighted chance constraints (WCCs) by assigning different weights to different magnitudes of constraint violations. We use a linear weight function and a piece-wise linear control policy in the WCC-ED model. The CC-ED model will serve as a benchmark for the proposed models in Section~\ref{Sec:Proposed_Models}, while the  WCC-ED model will provide insights for the proposed model in Section~\ref{Subsection:LDT-WCC}.

\subsection{CC-ED: Chance-Constrained Economic Dispatch Model} \label{Subsection:CC}

We consider wind power as the sole source of uncertainty and thermal generators as the sole flexible resource for managing wind power fluctuations. However, the model can be extended to include other uncertainty sources, such as electricity demand or contingencies, as well as other flexible resources, such as
energy storage.

We denote the uncertain output of wind farm $n'$ as $w_{n'} = \hat{w}_{n'} + \boldsymbol{\omega}_{n'}$, where the deterministic value $\hat{w}_{n'}$ is the forecasted power, and the random variable $\boldsymbol{\omega}_{n'}$ is the forecast error. At the system level, the aggregated wind power forecast is $\hat{W} = \sum_{n' \in \mathcal{W}} \hat{w}_{n'}$, and the aggregated forecast error is $\boldsymbol{\Omega} = \sum_{n' \in \mathcal{W}} \boldsymbol{\omega}_{n'}$. Accordingly, the output of generator $n$ can be modeled as $g_{n}(\boldsymbol{\Omega}) = p_n + \delta_{n}(\boldsymbol{\Omega})$, where the deterministic variable $p_n$ is the scheduled power generation under $\hat{W}$, and the random component $\delta_{n}(\boldsymbol{\Omega})$ is the reserved flexible capacity from generator $n$ to balance the uncertainty $\boldsymbol{\Omega}$. Following \cite{bienstock2014chance,dvorkin2019chance,kuang2018pricing}, we formulate the CC-ED model as:
\begin{subequations} \label{ModelCC}
    \begin{align} 
    \min_{p, g, \delta} ~& \mathbb{E}_{\boldsymbol{\boldsymbol{\Omega}}} \big[ \sum_{n\in\mathcal{N}} C_{n} (g_n(\boldsymbol{\Omega})) \big] \label{ModelCC:OF}\\
    \text{s.t.} \quad & ~ p_n \geq 0 &\forall n \label{ModelCC:pn}\\
    &~ g_{n}(\boldsymbol{\Omega}) = p_n + \delta_{n}(\boldsymbol{\Omega}) &\forall n \label{ModelCC:Policy}\\
    &~ \mathbb{P}_{\boldsymbol{\Omega}} \big[g_{n}(\boldsymbol{\Omega}) \leq p_{n}^{\max}\big] \geq 1 - \epsilon_n &\forall n \label{ModelCC:MaxLimit}\\
    &~ \mathbb{P}_{\boldsymbol{\Omega}} \big[ p_{n}^{\min} \leq g_{n}(\boldsymbol{\Omega})\big] \geq 1 - \epsilon_n &\forall n \label{ModelCC:MinLimit}\\
    & \sum_{n \in \mathcal{N}} p_n = D - \hat{W} \label{ModelCC:EnergyBalance}\\
    & \sum_{n \in \mathcal{N}} \delta_{n}(\boldsymbol{\Omega}) = \boldsymbol{\Omega}, \label{ModelCC:ReserveBalance}
\end{align}
\end{subequations}
where objective \eqref{ModelCC:OF} minimizes the expected power generation cost under uncertainty, and \eqref{ModelCC:pn} defines the feasible region of scheduled generation $p_n$. Equation  \eqref{ModelCC:Policy} models generator output under uncertainty, with generation limits imposed as CCs in \eqref{ModelCC:MaxLimit} and \eqref{ModelCC:MinLimit}, limiting the constraint violation rate for generator $n$ to $\epsilon_n$. Constraint \eqref{ModelCC:EnergyBalance} enforces power balance under the forecasted wind power $\hat W$, and \eqref{ModelCC:ReserveBalance} represents reserve deployment in response to wind power fluctuations. Function $\delta_n(\boldsymbol{\Omega})$ describes how generator $n$ responds to wind power fluctuations. In this model, we assume an affine control policy $\delta(\boldsymbol{\Omega}) = \alpha_n \boldsymbol{\Omega}$, where $\alpha_n \in [0,1]$ is the participation factor of generator $n$. Under this policy, \eqref{ModelCC:ReserveBalance} can be simplified to $\sum_{n\in \mathcal{N}} \alpha_n = 1$.

We assume the forecast error $\boldsymbol{\Omega}$ follows a Gaussian distribution, i.e., $\boldsymbol{\Omega} \sim \mathcal{N}(\mu_{\Omega},\sigma_{\Omega}^2)$, and $\mu_{\Omega} = 0$, meaning the forecast error does not have a systematic bias. By applying the convex reformulation method in \cite{bienstock2014chance}, the CCs in \eqref{ModelCC:MaxLimit} and \eqref{ModelCC:MinLimit} can reformulated as:
\begin{align} \label{CC_reform}
     p_{n}^{\min} + \alpha_n \hat{\sigma}_n \leq p_n \leq p_{n}^{\max} - \alpha_n \hat{\sigma}_n, & \quad \forall n
\end{align}
where $\hat{\sigma}_n= \Phi^{-1}(1-\epsilon_n) \sigma_{\Omega}$ is a given parameter, and $\Phi^{-1}(\cdot)$ is the inverse cumulative distribution. This CC-ED model will serve as a benchmark for the proposed models in Section~\ref{Sec:Proposed_Models}.

\subsection{From CCs to Weighted Chance Constraints (WCCs)} \label{Section:WCC}
The CC-ED model in \eqref{ModelCC} has two drawbacks. First, it assumes affine control policies, $g_{n}(\boldsymbol{\Omega}) = p_n + \alpha_n \boldsymbol{\Omega}$, which may not always be optimal. Second, the model does not distinguish between large and small constraint violations, even though they correspond to different levels of risk. As proposed in \cite{roald2015optimal}, 
we can extend the CCs in \eqref{ModelCC:MaxLimit} and \eqref{ModelCC:MinLimit} to weighted chance constraints (WCCs), which allow for more flexible control policies beyond affine ones and enable differentiated responses based on the magnitude of uncertainty realizations. 

A general WCC takes the following form:
\begin{equation} \label{WCC}
    \int_{-\infty}^{\infty} f(y(\boldsymbol{\Omega}))\mathcal{P}(\boldsymbol{\Omega}) d\boldsymbol{\Omega} \leq \epsilon,
\end{equation}
where $\mathcal{P}(\boldsymbol{\Omega})$ is the probabilistic distribution of the uncertain variable $\boldsymbol{\Omega}$, and $y(\boldsymbol{\Omega})$, referred to as the overload component, quantify the magnitude of the constraint violation. For example, the overload component for \eqref{ModelCC:MaxLimit} is $y_n(\boldsymbol{\Omega}) = g_n(\boldsymbol{\Omega}) - p_n^{\max}$, where $y_n>0$ indicates a violation of the maximum power limit for generator $n$, while $y_n \le 0$ implies a safe operating region. Finally, $f(\cdot)$ is a weight function that evaluates the risk related to the overload, so $f(y)$ is nonzero only when $y > 0$. If $f(y)$ is the unit step function $\chi(y> 0)$, i.e., $f(y)=0$ for $y < 0$ and $f(y)=1$ for $y \ge 0$, the WCC in \eqref{WCC} becomes a standard CC.

In this paper, we use the linear weight function $f(y) =y \chi(y > 0)$. Under affine policies, \eqref{WCC} can be reformulated as:
\begin{equation} \label{WCC-2}
    \int_{-\infty}^{\infty} y \chi(y > 0) \mathcal{P}(\boldsymbol{\Omega}) d\boldsymbol{\Omega} = \int_{0}^{\infty} y \mathcal{P}(y) dy \leq \epsilon.
\end{equation}
This reformulation changes the variable of integration from $\boldsymbol{\Omega}$ to $y$, which is also a random value. Since $\boldsymbol{\Omega}$ follows a Gaussian distribution and $y(\boldsymbol{\Omega})$ is a linear transformation of $\boldsymbol{\Omega}$, $y(\boldsymbol{\Omega})$ remains Gaussian, i.e., $y_n \sim \mathcal{N}(\Tilde{\mu}_n,\Tilde{\sigma}_n^2)$, where the mean $ \Tilde{\mu}_n$ and variance $\Tilde{\sigma}_n$ can be derived from $\mu_{\Omega}$ and $\sigma_{\Omega}$, as detailed in \cite{roald2015optimal}. Therefore, we can reformulate \eqref{WCC} using the expectation of a truncated Gaussian distribution:
\begin{align} \label{WCC-linear}
    \Tilde{\mu} \Big(1 \Minus \Phi\Big( \frac{\Minus\Tilde{\mu}}{\Tilde{\sigma}}\Big)\Big) + \frac{\Tilde{\sigma}}{\sqrt{2\pi}} e^{\frac{\Minus 1}{2} \big( \frac{\Minus\Tilde{\mu}}{\Tilde{\sigma}} \big)^2 }\leq \epsilon,
\end{align}

Accordingly, we can extend the CCs in \eqref{ModelCC:MaxLimit} and \eqref{ModelCC:MinLimit} to the following WCCs model:
\begin{subequations} \label{ModelWCC}
\begin{align}
     & \Tilde{\mu}_n^{\max} \big(1\Minus \Phi\big(z^{\max}_{n}\big)\big) \Plus \frac{\Tilde{\sigma}_n^{\max}}{\sqrt{2\pi}} e^{\frac{\Minus1}{2} \big(z^{\max}_{n}\big)^2}\leq \epsilon_n \label{ModelWCC:MaxLimit}\\
     &\Tilde{\mu}_n^{\min} \big(1\Minus\Phi\big( z^{\min}_{n}\big)\big) \Plus \frac{\Tilde{\sigma}_n^{\min}}{\sqrt{2\pi}} e^{ \frac{\Minus1}{2} \big(z^{\min}_{n}\big)^2} \leq \epsilon_n \label{ModelWCC:MinLimit}
\end{align}
\end{subequations}
where $z^{\max}_{n} = {\Minus\Tilde{\mu}_n^{\max}}/{\Tilde{\sigma}_n^{\max}}$ and $z^{\min}_{n} = {\Minus\Tilde{\mu}_n^{\min}}/{\Tilde{\sigma}_n^{\min}}$. 
With the assumption that $\mu_{\Omega} = 0$, we have $\Tilde{\mu}_n^{\max} = p_n - p_n^{\max}$, $\Tilde{\mu}_n^{\min} = p_n^{\min} - p_n$, and $(\Tilde{\sigma}_n^{\max})^2 = (\Tilde{\sigma}_n^{\min})^2 = \alpha_n^2 \sigma_{\Omega}^2$. Note that $\Tilde{\mu}_n^{\max}$, $\Tilde{\mu}_n^{\max}$, $\Tilde{\sigma}_n^{\max}$ and $\Tilde{\sigma}_n^{\min}$ are variables rather than parameters since they contain decision variables $p_n$ and $\alpha_n$.

The WCC in \eqref{WCC} also enables the incorporation of non-affine control policies, which are more realistic for power system operations. Since larger constraint violations typically carry greater weights, generators can respond to larger wind power fluctuations with greater intensity. As a result, the generator control policy can be modified from the original affine policy $g_{n}(\boldsymbol{\Omega}) = p_n + \delta_{n}(\boldsymbol{\Omega})$ to a piece-wise linear policy as follows:
\begin{equation} \label{piece-wise linear policy}
    {g}_{n}(\boldsymbol{\Omega})=
    \begin{cases}
        p_n + \delta_n^{-}(\boldsymbol{\Omega}) , ~\quad ~~ \boldsymbol{\Omega} \leq \Omega_{\epsilon} \\
        p_n + \delta_n^{+}(\boldsymbol{\Omega}) , ~\quad ~~ \Omega_{\epsilon}  < \boldsymbol{\Omega}
    \end{cases}
\end{equation}
where $\delta_n^{+}(\cdot)$ and $\delta_n^{-}(\cdot)$ are the control functions for each region of the piece-wise affine policy, and $\Omega_{\epsilon} $ is the threshold value of $\boldsymbol{\Omega}$ where the generator's control policy changes. Generally, $\delta_n^{+}(\cdot)$ has a steeper slope than $\delta_n^{-}(\cdot)$. Under the policy in \eqref{piece-wise linear policy}, the WCC for \eqref{ModelCC:MaxLimit} can be written as:
\begin{align} \label{WCC-piece-wise-policy}
    &\int_{-\infty}^{\Omega_{\epsilon} } \int_{0}^{\infty} y_n \mathcal{P}(y_n|\boldsymbol{\Omega}) \mathcal{P}(\boldsymbol{\Omega}) dy_n d\boldsymbol{\Omega} \nonumber \\
   &+ \int_{\Omega_{\epsilon} }^{\infty} \int_{0}^{\infty} y_n \mathcal{P}(y_n|\boldsymbol{\Omega})\mathcal{P}(\boldsymbol{\Omega}) dy_n d\boldsymbol{\Omega}
    \leq \epsilon_n,
\end{align}
where the random variable $y_n|\boldsymbol{\Omega}$ follows a normal distribution, with its mean and variance depending on which of the two regions in \eqref{piece-wise linear policy} the total wind deviation $\boldsymbol{\Omega}$ falls into. We will provide a more detailed discussion on how to quantify these mean and variance in Section~\ref{Subsection:LDT-WCC} when we introduce the LDT-WCC formulations.

A key feature of WCCs is that they remain convex under general control policies, provided the weight function $f(\cdot)$ is convex (Theorem 1, \cite{roald2015optimal}). This convexity enables the derivation of globally optimal control policies and pricing strategies based on a WCC model.

\section{Proposed Formulations under Extreme Events} 
\label{Sec:Proposed_Models}
Although the CC-ED model in \eqref{ModelCC} incorporates regular wind power uncertainty into the reserve scheduling process, it does not adequately capture the risk of large deviations. To address the need for distinguishing between small and large deviations and accounting for the probabilities of rare or extreme events, in Section~\ref{Subsection:LDT-CC}, we extend CC-ED to incorporate LDT-CCs for extreme reserve scheduling, resulting in the LDT-CC-ED model. In Section~\ref{Subsection:LDT-CC-solution}, we reformulate the LDT-CC-ED model into a computationally tractable single-level optimization problem. To reduce conservatism in extreme reserve scheduling and save operational costs, Section~\ref{Subsection:LDT-WCC} introduces an LDT-WCC-ED model as a relaxation of LDT-CC-ED. Finally, in Section~\ref{Subsection:LDT-WCC-solution}, we reformulate LDT-WCC-ED to a bi-linear optimization problem. We then present a numerical solution method based on the cutting-plane algorithm to iteratively solve the LDT-WCC-ED problem.

\subsection{LDT-CC-ED: ED with Large Deviation Theory CCs} 
\label{Subsection:LDT-CC}

The LDT-CC-ED model schedules two components of reserved capacity for each generator $n$, quantified by the participation factors $\alpha_n$ and $\beta_n$ in a linear control policy. These two components, referred to as the ``regular reserve'' and ``extreme reserve'' respectively, are used to hedge the risk associated with regular wind fluctuations and extreme/rare events. The LDT-CC-ED model is formulated as follows:
\begin{subequations} \label{LDT-General}
\begin{align}
\min_{p,\alpha,\beta} ~& \mathbb{E}_{\boldsymbol{\Omega}} \big[ \sum_{n \in \mathcal{N}} C_n(p_n,\alpha_n, \beta_n) \big]\\
\text{s.t.} \quad &\alpha_n,~ \beta_n,~ p_n \geq 0 & \forall n \label{LDT-General-Variables}\\
&\mathbb{P}_{\boldsymbol{\Omega}}[p_n +\alpha_n \boldsymbol{\Omega} \leq p^{\max}_n] \geq 1-\epsilon_n & \forall n \label{LDT-General:MaxLimitCC}\\
&\mathbb{P}_{\boldsymbol{\Omega}}[p^{\min}_n \leq p_n +\alpha_n \boldsymbol{\Omega}] \geq 1-\epsilon_n & \forall n \label{LDT-General:MinlimitCC}\\
&\mathbb{P}_{\boldsymbol{\Omega}_\Plus}[p_n + \delta_n(\boldsymbol{\Omega}_\Plus) \leq p^{\max}_n] \geq 1-\epsilon^{\text{ext}}_n & \forall n \label{LDT-General:MaxlimitLDT}\\
&\mathbb{P}_{\boldsymbol{\Omega}_\Minus}[p^{\min}_n \leq p_n + \delta_n(\boldsymbol{\Omega}_\Minus)] \geq 1-\epsilon^{\text{ext}}_n & \forall n \label{LDT-General:MinlimitLDT}\\
&\sum_{n \in \mathcal{N}} p_n = D - \hat{W} \label{LDT-General:EnergyBalance}\\
&\sum_{n \in \mathcal{N}} \alpha_n = 1 \label{LDT-General:RegularReserveBalance}\\
&\sum_{n \in \mathcal{N}} \beta_n = 1. \label{LDT-General:ExtremeReserveBalance}
\end{align}
\end{subequations}
Model \eqref{LDT-General} includes two sets of chance constraints, which differ in how they quantify risk for regular and extreme forecast deviations, as well as in the control policies associated with these deviations. Specifically, \eqref{LDT-General:MaxLimitCC}-\eqref{LDT-General:MinlimitCC} are regular CCs that schedule regular reserve with participation factor $\alpha_n$, in response to regular deviations within $\boldsymbol{\Omega}$. In contrast, \eqref{LDT-General:MaxlimitLDT}-\eqref{LDT-General:MinlimitLDT} are LDT-CCs, which schedule reserve with participation factor $\beta_n$ to handle extreme or rare events. The extreme event set $\boldsymbol{\Omega_\Plus}$ in \eqref{LDT-General:MaxlimitLDT} is defined as $\{ \boldsymbol{\Omega_\Plus}: p_n + \delta_n (\boldsymbol{\Omega_\Plus}) \geq p_n^{\max}, \ \forall n\}$, with its boundary corresponding to scenarios where all generators operate at their maximum output limits. Similarly, the other extreme event set $\boldsymbol{\Omega_\Minus}$ in \eqref{LDT-General:MinlimitLDT} is defined as $\{\boldsymbol{\Omega_\Minus} : p_n + \delta_n (\boldsymbol{\Omega_\Minus}) \leq p_n^{\min}, \ \forall n\}$. 
The acceptable violation rates for CCs and LDT-CCs are $\epsilon_n$ and $\epsilon_n^{\text{ext}} \in [0,1]$, respectively, where $\epsilon_n \ll 1$ and $\epsilon_n^{\text{ext}} \lll 1$. By applying the same reformulation steps used for \eqref{ModelCC:MaxLimit} and \eqref{ModelCC:MinLimit}, the CCs in \eqref{LDT-General:MaxLimitCC} and \eqref{LDT-General:MinlimitCC} can be transformed into \eqref{CC_reform}.

Similar to the CC-ED and WCC-ED models, LDT-CC-ED requires power balance as stated in \eqref{LDT-General:EnergyBalance} and the constraint on the adequacy of regular reserve as in \eqref{LDT-General:RegularReserveBalance}. Additionally, it enforces extreme reserve adequacy through \eqref{LDT-General:ExtremeReserveBalance}. Note that the regular reserve $\alpha_n \boldsymbol{\Omega}$ also contributes to satisfying the LDT-CCs in \eqref{LDT-General:MaxlimitLDT} and \eqref{LDT-General:MinlimitLDT}.
To avoid overlap between the provision of regular and extreme reserves, we define the control policy for extreme reserve as $\delta_n(\boldsymbol{\Omega}) = (\alpha_n - \beta_n)\hat{\sigma}_n + \beta_n\boldsymbol{\Omega}$, where $\hat{\sigma}_n= \Phi^{-1}(1-\epsilon_n) \sigma_{\Omega}$ is a given parameter.

\subsection{Reformulations of LDT-CC-ED} 
\label{Subsection:LDT-CC-solution}

Following \cite{tong2022optimization}, we can  reformulate \eqref{LDT-General:MaxlimitLDT} and \eqref{LDT-General:MinlimitLDT} into more tractable forms using the dominating points $\Omega^*_\Plus \in \boldsymbol{\Omega_\Plus}$ and $\Omega^*_\Minus \in \boldsymbol{\Omega_\Minus}$. Therefore, the LDT-CC-ED problem is essentially a bi-level optimization problem, where the upper-level problem solves \eqref{LDT-General} with given $\Omega^*_\Plus$ and $\Omega^*_\Minus$, and the lower-level problem  determines  dominating points $\Omega^*_\Plus$ and $\Omega^*_\Minus$. In the following, we will introduce the reformulation method for LDT-CC in two steps, based on  general probability $\mathbb{P}[F (p,\boldsymbol{\Omega}) \leq 0]$.

First, we will explain how to determine  dominating point $\Omega^*$ based on  rate function $I(\boldsymbol{\Omega})$. 
According to \cite{tong2022optimization,tong2021extreme}, lower values of $I(\boldsymbol{\Omega})$ indicate more probable events, while higher values correspond to less probable events. Therefore,  dominating point $\Omega^*$ is the point within the extreme event set that minimizes the rate function, i.e., $\Omega^*\in \text{arg}\min_{\boldsymbol{\Omega}} \ I(\boldsymbol{\Omega})$. As noted in \cite{tong2022optimization}, $I(\cdot)$ is the conjugate of the cumulant-generating function, providing an alternative to moments for characterizing the distribution. Consequently, $I(\boldsymbol{\Omega})$ can be determined based on an assumed distribution of $\boldsymbol{\Omega}$.

Next, we will explain how to estimate  probability $\mathbb{P}[F (p,\boldsymbol{\Omega}) \leq 0]$ with  given $\Omega^*$, where the challenge is the potential nonlinearity of $F$ with respect to $\boldsymbol{\Omega}$. Following \cite{tong2022optimization}, we use $F_k(p, \boldsymbol{\Omega}; \Omega^*)$, which is the $k$ th-order Taylor's approximation of $F (p,\boldsymbol{\Omega})$ at point $\Omega^*$, to approximate $\mathbb{P}[F (p,\boldsymbol{\Omega}) \leq 0]$, when $\boldsymbol{\Omega}$ follows a Gaussian distribution. We can then compute the first- and second-order probability estimates, $P_1$ and $P_2$ of $\mathbb{P}[F (p,\boldsymbol{\Omega}) \leq 0]$, by computing the measure of the sets bounded by the corresponding Taylor approximations:
$P_k(p, \boldsymbol{\Omega}) = P  (F_k(p, \boldsymbol{\Omega}; \Omega^*  ) \leq  0 ) , k \in \{1,2\}$. The specific expression of $F_k$ can be found in Section 3.1 in \cite{tong2022optimization}. This approach approximates the nonlinear (in $\boldsymbol{\Omega}$) chance
constrained problem with a linear and quadratic problem, respectively.

We can apply the two steps outlined above to the LDT-CCs in \eqref{LDT-General}. As an example, we will focus on reformulating \eqref{LDT-General:MaxlimitLDT}, with a similar approach applicable to \eqref{LDT-General:MinlimitLDT}. The reformulation of \eqref{LDT-General:MaxlimitLDT} is as follows:
\begin{subequations} \label{LDTraw}
\begin{align} 
    & P_1(y_n^{\Plus}(p_n,\alpha_n,\beta_n, \boldsymbol{\Omega}_\Plus; \Omega^*_\Plus)\ge 0) \leq \epsilon^{\text{ext}}_n, \ \forall n \label{LDTraw:Taylor}\\
    & \Omega^*_\Plus\in \text{arg}\min_{\boldsymbol{\Omega_\Plus}} \{ I(\boldsymbol{\Omega_\Plus}): p_n + \delta_n(\boldsymbol{\Omega_\Plus}) \geq p_n^{\max},\ \forall n \}, \label{LDTraw:LowerLevel}
\end{align} 
\end{subequations}
where $y_n^{\Plus} = p_n + \delta_n(\boldsymbol{\Omega}_\Plus) - p^{\max}_n$ is the overload component for \eqref{LDT-General:MaxlimitLDT}. Note that in \eqref{LDTraw:Taylor} we only use the first-order Taylor’s approximation since the critical region $\boldsymbol{\Omega_\Plus}$ has a linear boundary. 

Recall that we assume the wind power forecast errors follow a zero-mean Gaussian distribution, i.e. $\boldsymbol{\Omega} \sim \mathcal{N}(0,\sigma_{\Omega}^2)$. According to \cite{tong2022optimization}, under this assumption, the rate function can be computed using the Legendre transformation of the cumulant function as $I(\boldsymbol{\Omega}) = \frac{1}{2} \boldsymbol{\Omega}^{2} (\sigma_{\Omega}^2)^{\Minus 1}$. Using the first-order probability estimate $P_1(\cdot) = \Phi( \Minus\sqrt{2I(\Omega^*)})$ and substituting $\mu_{\Omega} = 0$, the LDT-CC-ED problem in \eqref{LDT-General} can be reformulated as:
\begin{subequations} \label{LDT-Reformulation}
\begin{align} 
\min_{p,\alpha,\beta} ~& \mathbb{E}_{\boldsymbol{\Omega}} \big[ \sum_{n \in \mathcal{N}} C_n(p_n,\alpha_n, \beta_n) \big]\\
\text{s.t.} \quad & \eqref{CC_reform}, \eqref{LDT-General-Variables}, \eqref{LDT-General:EnergyBalance}-\eqref{LDT-General:ExtremeReserveBalance} \nonumber \\
& \Phi\Big( \Minus \sqrt{{\Omega^*_\Plus}^{2} (\sigma_{\Omega}^2)^{\Minus 1}}\Big) \leq \epsilon^{\text{ext}}_n, \ \forall n\\
& \Omega^*_\Plus\in \text{arg}\min_{\boldsymbol{\Omega}} \Big\{ \frac{1}{2} \boldsymbol{\Omega}^{2} (\sigma_{\Omega}^2)^{\Minus 1}: p_n \nonumber \\
& \qquad \Plus (\alpha_n - \beta_n)\hat{\sigma}_{n} + \beta_n\boldsymbol{\Omega} \geq p_n^{\max} ,\ \forall n \Big\}. \label{LDT-Reformulation-LowerLevel}
\end{align}
\end{subequations}

Note that \eqref{LDT-Reformulation} is still a bi-level optimization problem which is hard to solve. However, we can transform \eqref{LDT-Reformulation} into a single-level problem by replacing \eqref{LDT-Reformulation-LowerLevel} with its first-order optimality conditions, resulting in the computationally tractable LDT-CC-ED model as follows:
\begin{subequations} \label{LDT-CC}
\begin{align}
\min_{p,\alpha,\beta, \Omega^*_\Plus, \lambda^*_\Plus} ~& \mathbb{E}_{\boldsymbol{\Omega}} \big[ \sum_{n \in \mathcal{N}} C_n(p_n,\alpha_n, \beta_n) \big]\\
\text{s.t.}\quad &\alpha_n, \beta_n, p_n \geq 0, \lambda^{\Plus *}_{n} > 0 & \forall n \label{LDT-CC-T0}\\
(\delta_{n}^{\Plus}): \quad & p_n - p_{n}^{\max} + \alpha_n \hat{\sigma}_n \leq 0 & \forall n \label{LDT-CC-T1}\\
(\mu_{n}^{\Plus}): \quad &  p_n + (\alpha_n - \beta_n)\hat{\sigma}_{n} + \beta_n\Omega^*_\Plus = p_n^{\max}  & \forall n \label{LDT-CC-R1}\\
(\kappa^{\Plus}): \quad & \Minus(\sigma_{\Omega}^2)^{\Minus 1/2} \Omega^*_\Plus - \Phi^{\Minus 1}( \epsilon^{\text{ext}}) \leq 0\\
(\xi_n^{\Plus}): \quad & (\sigma_{\Omega}^2)^{\Minus 1} \Omega^*_\Plus - \beta_n \lambda^{\Plus *}_{n} = 0 & \forall n \label{LDT-CC-R3}\\
(\pi): \quad &\sum_{n \in \mathcal{N}} p_n = D - \hat{W} \label{LDT-CC-demand}\\
(\rho): \quad &\sum_{n \in \mathcal{N}} \alpha_n = 1 \label{LDT-CC-reserve_op}\\
(\chi): \quad & \sum_{n \in \mathcal{N}} \beta_n = 1. \label{LDT-CC-reserve_ad}
\end{align}
\end{subequations}
Here, Greek letters in parentheses on the left denote dual multipliers of the corresponding constraints. The additional decision variable $\lambda^{\Plus *}$ is the optimal value of the dual multiplier associated with the generator limit constraints in the lower-level problem \eqref{LDT-Reformulation-LowerLevel}. Note that $\lambda^{\Plus *}$ has the same dimension as the number of generators in the system, which corresponds to the number of constraints in \eqref{LDT-Reformulation-LowerLevel}. Eqs. \eqref{LDT-CC-T1} is the reformulation of the regular CC in \eqref{LDT-General:MaxLimitCC}. Eqs. \eqref{LDT-CC-R1}-\eqref{LDT-CC-R3} are the reformulations of \eqref{LDT-General:MaxlimitLDT}. 
For simplicity, the reformulations of \eqref{LDT-General:MinlimitCC} and \eqref{LDT-General:MinlimitLDT} are not provided here, but they follow the same structure as the reformulations of their corresponding upper limit constraints.

In summary, LDT-CC offers a tractable formulation for the control policy function by using Taylor approximations to estimate the true probability of rare realizations of the uncertainty $\boldsymbol{\Omega}$. Following the proposed reformulation steps, the LDT-CC-ED problem is reduced to a single-level ED problem which includes Taylor's approximation of chance constraints in the neighborhood of $\Omega^*$. This formulation allows the use of off-the-shelf solvers, is independent of samples or scenarios, and provides an analytical expression for electricity pricing, which will be further studied in Section~\ref{Sec:Pricing}.

\subsection{LDT-WCC-ED: A Relaxation of LDT-CC-ED}
\label{Subsection:LDT-WCC}
Although \eqref{LDT-CC} offers an effective way to schedule extreme reserve, it often leads to overly conservative solutions for rare and extreme events, potentially exacerbating out-of-merit  dispatch and leading to greater costs. This conservatism arises because the LDT-CCs in \eqref{LDT-General:MinlimitLDT} and \eqref{LDT-General:MaxlimitLDT} are designed to schedule extreme reserve to hedge against risks associated with $\Omega^*_\Plus$ and $\Omega^*_\Minus$. However, $\Omega^*_\Plus$ and $\Omega^*_\Minus$ lie at the boundary of the system's critical operating region, while in reality, a power system rarely reaches its operational limits. As a result, although the extreme reserve based on this principle ensures exceptionally high system reliability, some scheduled reserve capacity is in fact never activated, leading to cost inefficiencies.

To mitigate this conservatism, we propose a relaxation of the LDT-CC-ED model that leverages the information from  dominant point $\Omega^*$ but no longer requires full hedging of the risk assocaited with $\Omega^*$. Instead, we employ the WCC approach to differentiate risks far from $\Omega^*$ and those closer to $\Omega^*$. For example, to hedge the overload risk of generators, we can use the following piece-wise affine control policy: 
\begin{equation} \label{LDT-WCC-policy}
    g_{n}(\boldsymbol{\Omega}, \Omega^*_\Plus) {\rm{=}}
\begin{cases}
p_n + \alpha_{n}\boldsymbol{\Omega} , \qquad \qquad \qquad \quad \boldsymbol{\Omega} \leq \Omega_{\epsilon}^\Plus \\
p_n + \beta_n \Omega^*_\Plus +(\alpha_{n} - \beta_n)\boldsymbol{\Omega} , \ \Omega_{\epsilon}^\Plus < \boldsymbol{\Omega}.
\end{cases}
\end{equation}
Similar to the definitions of regular and extreme reserves in the LDT-CC-ED framework, we refer to $\alpha_n \Omega$ and $\beta_n \Omega$ as the amounts of regular and extreme reserves provided by generator $n$, respectively. The threshold value $\Omega_{\epsilon}^{\Plus}$ is an adjustable parameter that determines the separation point between scheduling regular and extreme reserves. This scheduling model is referred to LDT-WCC-ED.

Under the piece-wise affine control policy in \eqref{LDT-WCC-policy}, the LDT-WCC-ED model can be formulated as follows:
\begin{subequations} \label{LDT-WCC-Model}
\begin{align}
\min_{p,\alpha,\beta} \quad & \mathbb{E}_{\boldsymbol{\Omega}} \big[ \sum_{n \in \mathcal{N}} C_n(p_n,\alpha_n, \beta_n) \big]\\
\text{s.t.}\quad & \eqref{LDT-CC-demand}-\eqref{LDT-CC-reserve_ad} \label{LDT-WCC-Model-balances} \nonumber \\
&~\forall n ~\Big\{p_n, \alpha_n, \beta_n  \geq 0 \\
(\nu_n^\Plus): & ~\int_{\Minus\infty}^{\Omega_{\epsilon}^\Plus } \int_{0}^{\infty} y_n^\Plus \mathcal{P}(y_n^\Plus|\boldsymbol{\Omega}) \mathcal{P}(\boldsymbol{\Omega}) dy_n^\Plus d\boldsymbol{\Omega} \nonumber\\
&~ + \int_{\Omega_{\epsilon}^\Plus }^{\infty} \int_{0}^{\infty} y_n^\Plus \mathcal{P}(y_n^\Plus|\boldsymbol{\Omega}) \mathcal{P}(\boldsymbol{\Omega}) dy_n^\Plus d\boldsymbol{\Omega}\leq \epsilon_n  \label{LDT-WCC-Model-Maxlimit}\\
(\nu_n^\Minus): &~ \int_{\Minus\infty}^{\Omega_{\epsilon}^\Minus } \int_{0}^{\infty} y_n^\Minus \mathcal{P}(y_n^\Minus|\boldsymbol{\Omega}) \mathcal{P}(\boldsymbol{\Omega}) dy_n^\Minus d\boldsymbol{\Omega} \nonumber\\
&~ + \int_{\Omega_{\epsilon}^\Minus }^{\infty} \int_{0}^{\infty} y_n^\Minus \mathcal{P}(y_n^\Minus|\boldsymbol{\Omega}) \mathcal{P}(\boldsymbol{\Omega})dy_n^\Minus d\boldsymbol{\Omega}\leq \epsilon_n. \Big\}\label{LDT-WCC-Model-Minlimit}
\end{align}
\end{subequations}
where $y_n^\Plus =g_n(\boldsymbol{\Omega}) - p_n^{\max}$ is the overload component for upper limit of generator $n$'s output, while $y_n^\Minus = p_n^{\min} - g_n(\boldsymbol{\Omega}) $ is the overload component for lower limit of generator $n$'s output.

Comparing the LDT-WCCs in \eqref{LDT-WCC-Model-Maxlimit}-\eqref{LDT-WCC-Model-Minlimit} with the CCs in \eqref{ModelCC:MaxLimit}-\eqref{ModelCC:MinLimit}, we notice that they use the same risk tolerance $\epsilon_n$, meaning they hedge risk with the same cumulative probability. However, LDT-WCC assigns smaller weights to violations of smaller magnitude, it can generally cover a larger range of $y(\boldsymbol{\Omega})$ than the CCs, provided the weight function is designed appropriately.

On the other hand,  when comparing the LDT-WCCs with the LDT-CCs in \eqref{LDT-General:MaxlimitLDT}-\eqref{LDT-General:MinlimitLDT}, we notice that the explicit value of $\epsilon_n^{\text{ext}}$ is no longer present in the LDT-WCCs. However, it is implicitly embedded in the values of $\Omega^*_\Plus$ and $\Omega^*_{\Minus}$. Therefore, the LDT-WCC model  leverages anticipated extreme realizations to enhance system performance. We refer to this ability to exploit a priori rare event statistics for possible future extreme event realizations as ``anticipative preparedness''.

Fig.~\ref{fig:scheme} illustrates the differences between the CC benchmark, the proposed LDT-CC model, and the relaxed LDT-WCC model. In Fig.~\ref{fig:scheme}, the CC model hedges against deviations during regular operations with a cumulative probability of $(1-\epsilon)$, applying a uniform weight function for constraint violations. Both LDT-CC and LDT-WCC account for the characteristics of the tail distribution. However, LDT-CC covers deviations up to the dominating point $\Omega^*$, representing extreme conditions, while LDT-WCC prepares the system for these extreme conditions by incorporating information about $\Omega^*$ into the control policy and applying a linear weighting to constraint violations. In summary, when comparing the risk of constraint violation, the CC model covers the smallest range, followed by the LDT-WCC model, with the LDT-CC model covering the largest range.

The conservatism level of the LDT-WCC-ED can be adjusted by modifying the threshold $\Omega_{\epsilon}^\Plus$ in the generator control policy. Here, the conservatism level reflects the system operator's subjective perception of extreme scenarios. A larger value of  $\Omega_{\epsilon}^\Plus$ implies that the system operator considers only very large deviations from a given forecast as extreme and is willing to schedule extreme reserve only when the uncertainty is sufficiently high. To enable extreme reserve scheduling, it is necessary to ensure $\Omega_{\epsilon}^{\Plus} < \Omega^*_\Plus$, because when $\boldsymbol{\Omega} > \Omega^*_\Plus$ none of the generators has additional available capacity. If $\Omega_{\epsilon}^{\Plus} = \Omega^*_\Plus$, then all generators can provide only regular reserve, as they would always follow the first segment of the piecewise linear control policy in \eqref{LDT-WCC-policy}. However, adjusting this parameter not only influences reserve scheduling but also affects the prices of different reserve services (as discussed in Section~\ref{Sec:Pricing}), leading to a nonlinear impact on the system's overall operational cost. In real power systems, system and market operators can tune this values based on their experience and subjective trade-off between system reliability and operational cost.

\begin{figure}[!t]
    \centering
    \includegraphics[width=0.49\textwidth]{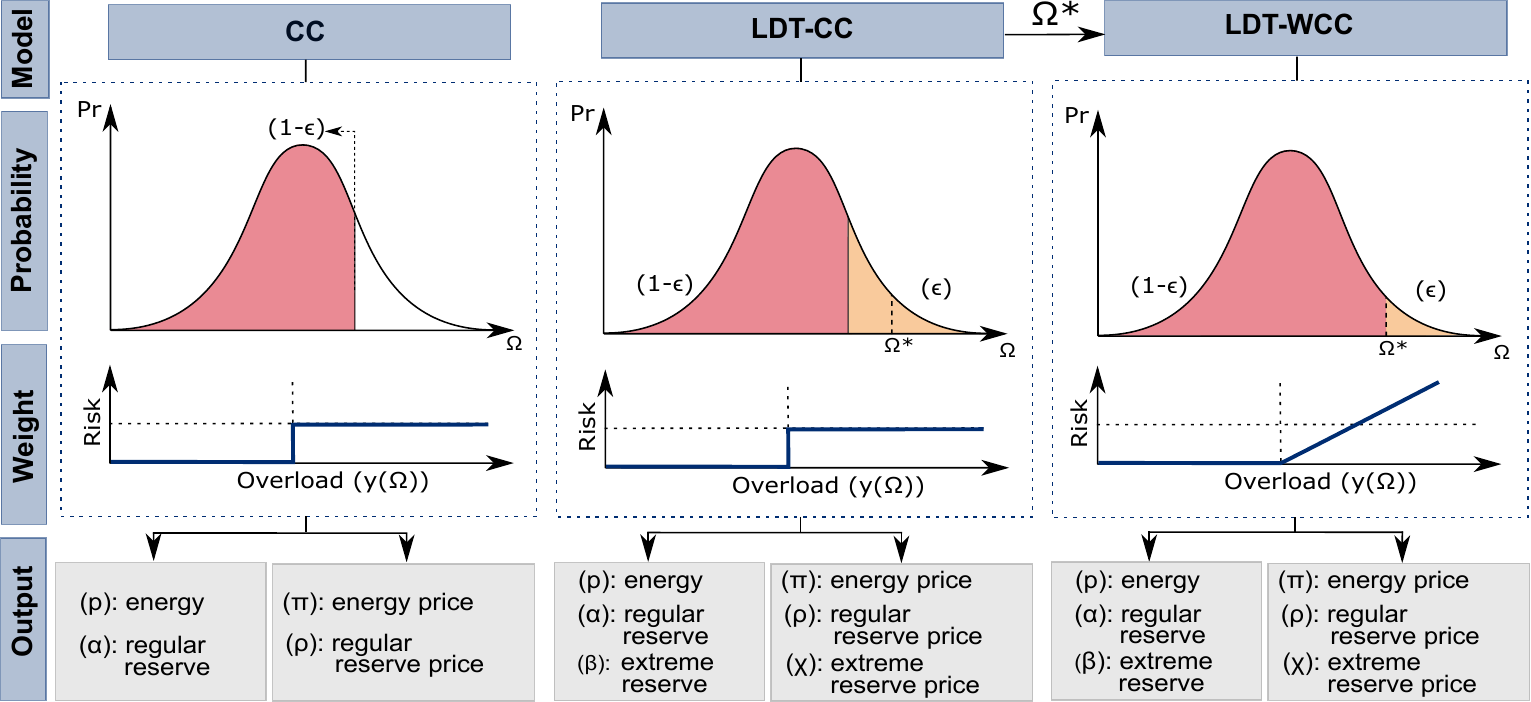}
    \caption{Comparison of benchmark (CC), the proposed model (LDT-CC), and its relaxation (LDT-WCC).}
\label{fig:scheme}
\end{figure}

\subsection{Reformulations of LDT-WCC-ED} 
\label{Subsection:LDT-WCC-solution}

Following a similar reformulation method as the standard WCC introduced in \cite{roald2015optimal}, the LDT-WCCs in \eqref{WCC-linear} can be transformed into the form of a truncated Gaussian distribution. In the following, we will focus on reformulating \eqref{LDT-WCC-Model-Maxlimit}, with a similar approach applicable to \eqref{LDT-WCC-Model-Minlimit}. The reformulation of \eqref{LDT-WCC-Model-Maxlimit} is as follows:
\begin{align}  \label{LDT-WCC-Reformulation}
     & \Tilde{\mu}_n^\text{I} \big(1- \Phi(z_n^\text{I})\big) + \frac{\Tilde{\sigma}_n^\text{I}}{\sqrt{2\pi}} e^{\frac{\Minus1}{2} (z_n^\text{I})^2} \nonumber \\
     &+ \Tilde{\mu}_n^\text{II} \big(1-\Phi (z_n^\text{II})\big) + \frac{\Tilde{\sigma}_n^\text{II}}{\sqrt{2\pi}} e^{ \frac{\Minus1}{2} (z_n^\text{II})^2} \leq \epsilon_n 
\end{align}
where $z^\text{I}_{n} = {\Minus\Tilde{\mu}_n^\text{I}}/{\Tilde{\sigma}_n^\text{I}}$, $z^\text{II}_{n} = {\Minus\Tilde{\mu}_n^\text{II}}/{\Tilde{\sigma}_n^\text{II}}$, $\Tilde{\mu}_n^\text{I}$ and $\Tilde{\mu}_n^\text{II}$ are the means of the random variable $y_n^\Plus|\boldsymbol{\Omega}$ when $\boldsymbol{\Omega} \leq \Omega_{\epsilon}^{\Plus} $ and $\boldsymbol{\Omega} > \Omega_{\epsilon}^{\Plus} $, respectively. Similarly, $(\Tilde{\mu}_n^\text{I})^2$ and $(\Tilde{\mu}_n^\text{II})^2$ are the variance of $y_n^\Plus|\boldsymbol{\Omega}$ in these two regions, respectively. The means are:
\begin{subequations} \label{LDT-WCC-Mean}
\begin{align} 
    \Tilde{\mu}_n^\text{I} & =  p_n - p_n^{\max} - \alpha_{n}\sigma_{\Omega} \frac{\phi(z_{\Omega*})}{\Phi(z_{\Omega*})} \\
    \Tilde{\mu}_n^\text{II} & = p_n - p_n^{\max} + \beta_n \Omega^* + (\alpha_{n}\Minus\beta_{n})\sigma_{\Omega} \frac{\phi(z_{\Omega*})}{1-\Phi(z_{\Omega*})}
\end{align}
\end{subequations}
and the variance are:
\begin{subequations} \label{LDT-WCC-Variance}
\begin{align}
(\Tilde{\sigma}_n^\text{I})^2 & =(\alpha_n \sigma_{\Omega})^2 \Big(1 - z_{\Omega*} \frac{\phi(z_{\Omega*})}{\Phi(z_{\Omega*})} - \Big(\frac{\phi(z_{\Omega*})}{\Phi(z_{\Omega*})}\Big)^2 \Big) \\
(\Tilde{\sigma}_n^\text{I})^2 & = \Big((\alpha_n - \beta_n) \sigma_{\Omega}\Big)^2 \Big( 1 + z_{\Omega*} \frac{\phi(z_{\Omega*})}{1 - \Phi(z_{\Omega*})} \nonumber \\
& \qquad \qquad \qquad \qquad- \Big(\frac{\phi(z_{\Omega*})}{1-\Phi(z_{\Omega*})} \Big)^2  \Big)
\end{align}
\end{subequations}
where $z_{\Omega*} = {\Omega^*}/{\sigma_\Omega}$, $\Phi(\cdot)$ is the cumulative distribution, and $\phi(\cdot)$ is a probability density function of a Gaussian distribution. Note that the mean and variance of $y_n^\Plus|\boldsymbol{\Omega}$ are variables rather than parameters because they contain $p_n$, $\alpha_n$ and $\beta_n$. 

The LDT-WCC formulation in \eqref{LDT-WCC-Reformulation} involves decision variables within the cumulative distribution function that are multiplied by other decision variables, necessitating an additional procedure for solvers to recognize it. To address this efficiently, we apply a cutting-plane method. Specifically, we iteratively solve a sequence of relaxations of \eqref{LDT-WCC-Model} without the original LDT-WCCs in \eqref{LDT-WCC-Model-Maxlimit} and \eqref{LDT-WCC-Model-Minlimit}, but with a set of cutting-plane constraints. These cutting planes are the first-order approximations of \eqref{LDT-WCC-Model-Maxlimit} and \eqref{LDT-WCC-Model-Minlimit} at different points.
This results in a linear optimization problem that can be handled by commercial solvers. 

Algorithm \ref{Cutting plane approach} outlines this cutting-plane method for solving \eqref{LDT-WCC-Model}. At each iteration $t$, we verify whether \eqref{LDT-WCC-Model-Maxlimit} and \eqref{LDT-WCC-Model-Minlimit} are satisfied by the current solution $p_n^t, \alpha_n^t, \beta_n^t$. If they are, the algorithm terminates. If not, we add linearizations of \eqref{LDT-WCC-Model-Maxlimit} and \eqref{LDT-WCC-Model-Minlimit} and resolve the problem in the next iteration. 

\begin{algorithm}
\caption{Cutting-Plane Algorithm for LDT-WCC-ED}
\label{Cutting plane approach}
\label{alg:example}
\begin{algorithmic}[1]
\State Initialize iteration $t=0$
\State Initialize an empty cutting-plane constraint set
\State Solve \eqref{LDT-WCC-Model} with the cutting-plane constraints, but without \eqref{LDT-WCC-Model-Maxlimit} and \eqref{LDT-WCC-Model-Minlimit}, to obtain ($p_n^t, \alpha_n^t, \beta_n^t$).
\State Check if ($p_n^t, \alpha_n^t, \beta_n^t$) satisfies \eqref{LDT-WCC-Model-Maxlimit} and \eqref{LDT-WCC-Model-Minlimit}.
\If {\eqref{LDT-WCC-Model-Maxlimit} or \eqref{LDT-WCC-Model-Minlimit} \textbf{is not} satisfied by ($p_n^t, \alpha_n^t, \beta_n^t$)}
    \State Compute $\Tilde{\mu}_n^t$ and $\Tilde{\sigma}_n^t$ following \eqref{LDT-WCC-Mean} and \eqref{LDT-WCC-Variance}.
    \State Add the first-order Taylor approximation of the unsatisfied constraint at ($\Tilde{\mu}_n^t$, $\Tilde{\sigma}_n^t$) to the cutting-plane constraint set.
    \State $t=t+1$
    \State Go to \textbf{Step 3}
\Else
    \State Exit
\EndIf
\end{algorithmic}
\end{algorithm}

For clarity, we demonstrate the linearization of \eqref{LDT-WCC-Model-Maxlimit} and \eqref{LDT-WCC-Model-Minlimit} using the basic element of a WCC reformulation in its simplest form, as follows:
\begin{align} \label{unit WCC}
    f(\Tilde{\mu}_n,\Tilde{\sigma}_n) =~& \Tilde{\mu}_n \Big( 1 - \Phi(z_n)\Big) + \frac{\Tilde{\sigma}_n}{\sqrt{2\pi}} e^{\frac{\Minus1}{2} (z_n)^2}.
\end{align}
Based on $p_n^t, \alpha_n^t, \beta_n^t$ which are the optimal solutions of the relaxed LDT-WCC-ED at iteration $t$, we can compute the values of $\Tilde{\mu}_n^t$ and $\Tilde{\sigma}_n^t$ following \eqref{LDT-WCC-Mean} and \eqref{LDT-WCC-Variance}. The first-order approximation of \eqref{unit WCC} at ($\Tilde{\mu}_n^t$, $\Tilde{\sigma}_n^t$) is:
\begin{align}
    & f(\Tilde{\mu}_n^t,\Tilde{\sigma}_n^t)+ \frac{\partial f(\cdot)}{\partial \Tilde{\mu}_n} \frac{\partial \Tilde{\mu}_n}{\partial p_n} (p_n - p_n^t) \nonumber \\
    & + \frac{\partial f(\cdot)}{\partial \Tilde{\mu}_n} \frac{\partial \Tilde{\mu}_n}{\partial \alpha_n} (\alpha_n \Minus \alpha_n^t) + \frac{\partial f(\cdot)}{\partial \Tilde{\mu}_n} \frac{\partial \Tilde{\mu}_n}{\partial \beta_n} (\beta_n \Minus \beta_n^t) \nonumber \\
    & + \frac{\partial f(\cdot)}{\partial \Tilde{\sigma}_n} \frac{\partial \Tilde{\sigma}_n}{\partial \alpha_n} (\alpha_n \Minus \alpha_n^t) + \frac{\partial f(\cdot)}{\partial \Tilde{\sigma}_n} \frac{\partial \Tilde{\sigma}_n}{\partial \beta_n} (\beta_n \Minus \beta_n^t) \leq \epsilon_n, \label{CP-Linearization}
\end{align}
where the derivatives with respect to $\Tilde{\mu}$ and $\Tilde{\sigma}$ are:
\begin{align}
    \frac{\partial f(\cdot)}{\partial \Tilde{\mu}_n} =& \Big(1 - \Phi\big(z_n\big)\Big) +\Tilde{\mu}_n \Big(\frac{\Minus 1}{\Tilde{\sigma}_n}\phi\big(z_n\big)\Big) - \frac{z_n}{\sqrt{2\pi}} e^{ \frac{\Minus 1}{2} (z_n)^2} \label{df_dmu}\\
    \frac{\partial f(\cdot)}{\partial \Tilde{\sigma}_n}=& -\frac{\Tilde{\mu}_n^2}{\Tilde{\sigma}_n^2} \phi\big(z_n\big) +\frac{1}{\sqrt{2\pi}} e^{ \frac{\Minus1}{2} (z_n)^2 } \Big(1 + \frac{\Tilde{\mu}_n^2}{\Tilde{\sigma}_n^2}\Big) \label{df_dtheta}.
\end{align}
In addition, we can compute other partial derivatives, including $\frac{\partial \Tilde{\mu}_n}{\partial p_n}$, $\frac{\partial \Tilde{\mu}_n}{\partial \alpha_n}$, $\frac{\partial \Tilde{\mu}_n}{\partial \beta_n}$, $\frac{\partial \Tilde{\sigma}_n}{\partial \alpha_n}$, $\frac{\partial \Tilde{\sigma}_n}{\partial \beta_n}$, based on the expressions for $\Tilde{\mu}_n$ in \eqref{LDT-WCC-Mean} and $\Tilde{\sigma}_n$ in \eqref{LDT-WCC-Variance}. The values of these partial derivatives depend on the region in which the realization of $\boldsymbol{\Omega}$ is located. 

Although the cutting-plane algorithm typically does not guarantee polynomial-time convergence \cite{lubin2015robust}, the convexity of the proposed models, combined with the compactness and non-emptiness of the feasible region, ensures convergence to the optimal solution \cite{bienstock2014chance}.

\section{Prices for Energy, Regular Reserve, and Extreme Reserve} \label{Sec:Pricing}
In this section, we derive the prices for providing energy, regular reserve, and extreme reserve based on the LDT-CC-ED model in Section~\ref{Subsection:LDT-CC-pricing} and the LDT-WCC-ED model in Section~\ref{Subsection:LDT-WCC-pricing}. We examine the competitive equilibrium established by the proposed market clearing models, where all market participants maximize their profits or utilities under the market clearing results, ensuring no incentive to deviate from the market outcomes. Additionally, we analyze key market properties, such as cost recovery and revenue adequacy, based on these equilibrium prices in Section~\ref{Subsection:market properties}. For notation simplicity, in this section we have omitted all constraints related to $p_n^{\min}$ and focused only on the CC, LDT-CC, and LDT-WCC constraints associated with $p_n^{\max}$. This can be interpreted as assuming $p_n^{\min} =0, \ \forall n \in \mathcal{N}$, thereby turning the lower power limit on generators  into hard constraints, i.e., $p_n \geq 0, \ \forall n \in \mathcal{N}$.

\subsection{Market-Clearing Based on LDT-CC-ED}
\label{Subsection:LDT-CC-pricing}
Consider the LDT-CC-ED model in \eqref{LDT-CC}. Let $\pi$, $\rho$, and $\chi$ denote the dual multipliers for the power balance constraint in \eqref{LDT-CC-demand}, the regular reserve sufficiency constraint in \eqref{LDT-CC-reserve_op}, and the extreme reserve sufficiency constraint in \eqref{LDT-CC-reserve_ad}, respectively. These multipliers represent the prices for energy, regular reserve, and extreme reserve. We begin by deriving these prices based on \eqref{LDT-CC}, as presented in the following proposition:
\begin{proposition} \label{prop1}
(Pricing based on LDT-CC-ED) The optimal values of $\pi$, $\rho$, and $\chi$ based on \eqref{LDT-CC} are given by:
    \begin{subequations} \label{LDT-CC prices}
    \begin{align}
        \pi~& = \frac{\partial C_n(\cdot)}{\partial p_n} + \mu_{n}^{\Plus} + \delta_{n}^{\Plus} \label{LDT-CC-KKT-pi} \\
        \rho ~&= \frac{\partial C_n(\cdot)}{\partial \alpha_n} + \mu_{n}^{\Plus} \hat{\sigma}_{n} + \delta_{n}^{\Plus} \hat{\sigma}_{n}\\
        \chi ~&= \frac{\partial C_n(\cdot)}{\partial \beta_n} +\mu_{n}^{\Plus} (\Omega^*_\Plus-\hat{\sigma}_{n}) - \xi_{n}^{\Plus} \lambda^{\Plus *}_{n}
    \end{align}
    \end{subequations}   
\end{proposition}

\begin{proof} 
The Lagrangian function for \eqref{LDT-CC} is: 
\begin{align}
    \mathcal{L} &= \sum_{n \in \mathcal{N}} C_n(p_n,\alpha_n,\beta_n) + \sum_{n \in \mathcal{N}} \delta_n^{\Plus}(p_n - p_n^{\max} + \alpha_n \hat{\sigma}_n) \nonumber\\
    &+ \sum_{n \in \mathcal{N}} \mu_{n}^{\Plus} \Big(p_n +(\alpha_n - \beta_n)\hat{\sigma}_n + \beta_n \Omega^*_\Plus - p_n^{\max}\Big) \nonumber\\
    &+ \kappa^{\Plus} \Big(\Minus (\sigma_{\Omega}^2)^{\Minus 1/2} \Omega^*_\Plus - \Phi^{\Minus 1} (\epsilon_n^{\text{ext}}) \Big) \nonumber\\
    &+\sum_{n \in \mathcal{N}} \xi_n^{\Plus} \Big( (\sigma_{\Omega}^2)^{\Minus 1} \Omega^*_\Plus - \beta_n\lambda^{\Plus *}_n \Big) + \pi \Big( D + \hat{W}  - \sum_{n\in\mathcal{N}} p_n \Big) \nonumber \\
    & + \rho \Big(1 - \sum_{n \in \mathcal{N}} \alpha_n\Big) + \chi \Big(1 - \sum_{n \in \mathcal{N}} \beta_n\Big).
\end{align}
Then, we can derive the stationary conditions for \eqref{LDT-CC} as:
    \begin{subequations} \label{LDT-KKT}
    \begin{align}
        \frac{\partial \mathcal{L}}{\partial p_n}:~& \frac{\partial C_n(\cdot)}{\partial p_n} + \mu_{n}^{\Plus} + \delta_{n}^{\Plus} - \pi = 0 \label{LDT-KKT-C1}\\
        \frac{\partial \mathcal{L}}{\partial \alpha_n}:~&\frac{\partial C_n(\cdot)}{\partial \alpha_n} + \mu_{n}^{\Plus} \hat{\sigma}_{n}  + \delta_{n}^{\Plus} \hat{\sigma}_{n} -\rho = 0 \label{LDT-KKT-C2}\\
        \frac{\partial \mathcal{L}}{\partial \beta_n}:~&\frac{\partial C_n(\cdot)}{\partial \beta_n} + \mu_{n}^{\Plus} (\Omega^*_\Plus-\hat{\sigma}_{n} ) - \xi^{\Plus}_n \lambda^{\Plus *}_{n} -\chi = 0 \label{LDT-KKT-C3}
    \end{align}
    \end{subequations}
Thus, $\pi$, $\rho$, and $\chi$ can be expressed directly from \eqref{LDT-KKT}.
\end{proof}

Next, we analyze how the prices derived from LDT-CC-ED lead to a competitive equilibrium problem:
\begin{theorem} \label{theo1}
   (Market equilibrium based on LDT-CC-ED) Let $\{p_n^*, \alpha_n^*, \beta_n^*, \Omega^*_\Plus, \lambda^{\Plus *}_n\}$ denote the optimal solution to \eqref{LDT-CC} and let $\{ \pi^*, \rho^*, \chi^* \}$ be the corresponding dual variables. Then $\{\{p_n^*, \alpha_n^*, \beta_n^*, \lambda^{\Plus *}_n ~\forall n\}, \Omega^*_\Plus, \pi^*, \rho^*, \chi^*\}$ constitutes a market equilibrium, i.e.:
   \begin{itemize}
        \item The market clears at $\sum p_n^* = D- \hat{W}$, $\sum \alpha_n^* =1$, and $\sum \beta_n^* = 1$ 
        \item Each producer maximizes its profit under the payment $\Gamma_n = \pi^* p_n^* + \rho^* \alpha_n^* + \chi^* \beta_n^*$
    \end{itemize}
\end{theorem}

\begin{proof}
    Given an optimal pair $(\Omega^*_\Plus, \lambda^{\Plus }_n)$, if $\{p_n^*, \alpha_n^*, \beta_n^*, ~\forall n\}$ is feasible and solved to optimality, then the optimal values $\{p_n^*, \alpha_n^*, \beta_n^*~\forall n\}$ must satisfy the equality constraints in \eqref{LDT-CC}. Consequently, we have $\sum p_n^* = D- \hat{W}$, $\sum \alpha_n^* =1$, and $\sum \beta_n^* = 1$.
    
    Next, we model each producer as a risk-neutral, profit-maximizing entity. Given the optimal values of $(\Omega^*_\Plus, \lambda^{\Plus *}_n)$ from \eqref{LDT-CC}, generator $n$ optimizes its market participation based on the following model:
    \begin{subequations} \label{LDT-KKT-firm}
        \begin{align}
        \max_{p_n,\alpha_n,\beta_n} ~&  - C_n(p_n,\alpha_n, \beta_n) + l_n p_n + b_n \alpha_n + w_n \beta_n\\
        \text{s.t.}\quad & p_n, \alpha_n, \beta_n \geq 0 \\
        (\hat{\mu}_n^\Plus):\quad & \Minus p_n^{\max} + p_n + (\alpha_n - \beta_n) \hat{\sigma}_{n} + \beta_n \Omega^*_\Plus = 0\\
        (\hat{\delta}_n^\Plus):\quad & p_n - p_{n}^{\max} + \alpha_n \hat{\sigma}_n \leq 0 \\
        (\hat{\xi}_n^{\Plus}):\quad & (\sigma_{\Omega}^2)^{\Minus 1} \Omega^*_\Plus - \beta_n \lambda_n^{\Plus *} = 0
    \end{align}
    \end{subequations}
    where $l_n$, $b_n$, $w_n$ are the prices for energy, regular reserve and extreme reserve for this generator. The Karush-Kuhn-Tucker (KKT) optimality conditions of \eqref{LDT-KKT-firm} are:
    \begin{subequations} 
    \begin{align}
        & \Minus l_n + \hat{\mu}_n^\Plus + \hat{\delta}_n^\Plus + \frac{\partial C_n(\cdot)}{\partial p_n} = 0\\
        & \Minus b_n + \hat{\mu}_n^\Plus \Omega^*_\Plus + \hat{\delta}_n^\Plus \hat{\sigma}_n + \frac{\partial C_n(\cdot)}{\partial \alpha_n} = 0\\
        & \Minus w_n + \hat{\mu}_n^+ (\Omega^*_\Plus -\hat{\sigma}_{n}) - \hat{\xi}^\Plus_n \lambda_n^{\Plus *} + \frac{\partial C_n(\cdot)}{\partial \beta_n} = 0\\
        & \Big(\Minus p_n^{\max} + p_n + (\alpha_n - \beta_n)\hat{\sigma}_{n}+ \beta_n \Omega^*_\Plus \Big) \perp \hat{\mu}_n^\Plus = 0 \\
        & (p_n - p_{n}^{\max} + \alpha_n \hat{\sigma}_n) \perp \hat{\delta}_n^\Plus = 0 \\
        & \Big((\sigma_{\Omega}^2)^{\Minus 1} \Omega^*_\Plus -  \beta_n \lambda_n^{\Plus *}\Big) \perp \hat{\xi}_n^\Plus  = 0\\
        &  \hat{\delta}_n^\Plus \geq 0, ~ \hat{\mu}^\Plus_n\in \mathbb{R},~ \hat{\xi}^\Plus_n \in \mathbb{R}
    \end{align}
    \end{subequations}
    Based on these KKT conditions, we can derive $l_n,b_n,w_n$ as:
    \begin{subequations} \label{LDT-CC prices individual}
        \begin{align}
        & l_n = \frac{\partial C_n(\cdot)}{\partial p_n} + \hat{\mu}_n^\Plus + \hat{\delta}_n^\Plus \\
        & b_n = \frac{\partial C_n(\cdot)}{\partial \alpha_n} + \hat{\mu}_n^\Plus \Omega^*_\Plus + \hat{\delta}_n^\Plus \hat{\sigma}_n\\
        & w_n = \frac{\partial C_n(\cdot)}{\partial \beta_n} +\hat{\mu}_n^\Plus (\Omega^*_\Plus-\hat{\sigma}_{n}) - \hat{\xi}_n^\Plus \lambda_n^{\Plus *}
    \end{align}
    \end{subequations}
    By comparing the prices in \eqref{LDT-CC prices} with those in \eqref{LDT-CC prices individual}, we find that $\pi = l_n$, $\rho = b_n$, $\chi =w_n, \forall n \in \mathcal{N}$. Thus, the market clears at prices $(\pi, \rho, \chi)$ that enable each producer to maximize their profit.
\end{proof}
To better understand the components of price, we consider a specific production cost function for generator $n$:
$C_n(p_n,\alpha_n,\beta_n) = C_{1,n} (p_n + \boldsymbol{\Omega}\alpha_n) + C_{2,n} (p_n + \boldsymbol{\Omega} \alpha_n)^2 +  C^{\beta}_{n} \beta_n$. 
Then, each term of $\mathbb{E}_{\boldsymbol{\Omega}}[C_n(\cdot)]$ can be explicitly expressed as:
\begin{subequations}
    \begin{align}
    \mathbb{E}_{\boldsymbol{\Omega}}[C_{1,n} (p_n \Plus \boldsymbol{\Omega}\alpha_n)]& = C_{1,n} (p_n + \mu_{\Omega}\alpha_n)\\ 
    \mathbb{E}_{\boldsymbol{\Omega}}[C_{2,n} (p_n \Plus \boldsymbol{\Omega}\alpha_n)^2]& = C_{2,n} (p_n^2 + 2\mu_{\Omega}p_n + \sigma_{\Omega}^2\alpha_n^2)\\
    \mathbb{E}_{\boldsymbol{\Omega}}[C_{n}^{\beta} \beta_n]& = C_{n}^{\beta} \beta_n
    \end{align}
\end{subequations}
Considering $\mu_{\Omega}=0$, the expected total system cost is:
\begin{align}  
    \mathbb{E}_{\boldsymbol{\Omega}} \Big[ \sum_{n \in \mathcal{N}} C(p_n,\alpha_n,\beta_n) \Big] & \Equal \sum_{n\in \mathcal{N}} \big( C_{2,n} (p_n^2 + \sigma_{\Omega}^2 \alpha_n^2)\nonumber\\
    & \qquad + C_{1,n} p_n + C^{\beta}_{n} \beta_n \big) \label{cost-structure}
\end{align}
Based on this specific cost function, the price expressions in \eqref{LDT-CC prices} become:
\begin{subequations} \label{LDT-CC specific prices}
\begin{align}
    \pi & = C_{1,n} + 2 C_{2,n}p_n + \mu_n^{\Plus} + \delta_n^{\Plus} \label{LDT-CC-energy-price} \\
    \rho & = 2 C_{2,n} \alpha_n + (\mu^{\Plus}_n + \delta_n^{\Plus} )\hat{\sigma}_n \label{LDT-CC-regular-price} \\
    \chi &= C^{\beta}_n  + \mu^{\Plus}_n (\Omega^*_\Plus - \hat{\sigma}_n) - \xi_{n}^{\Plus} \lambda^{\Plus *}_{n} \label{LDT-CC-extreme-price}
\end{align}
\end{subequations} 

We can further analyze the factors influencing prices based on the expressions in \eqref{LDT-CC specific prices}. Both  energy price $\pi$ and  regular reserve price $\rho$ include $\delta_{n}^{\Plus}$ and $\mu_{n}^{\Plus}$, which are the dual multipliers of the regular CC in \eqref{LDT-CC-T1} and the extreme CC in \eqref{LDT-CC-R1}. This indicates that a generator’s decision to provide extreme reserves also impacts the prices for energy and regular reserves. In contrast, the extreme reserve price $\chi$ is unaffected by the binding status of the regular CC in \eqref{LDT-CC-T1} but is influenced by the binding status of the extreme CCs in \eqref{LDT-CC-R1} and \eqref{LDT-CC-R3}.

Additionally, we observe that $\pi$ is independent of the uncertainty and risk parameters, while $\rho$ internalizes the variance of uncertainty. Furthermore, $\chi$ captures both $\sigma_{\Omega}$, which characterizes normal deviations, and $\Omega^*_\Plus$, which characterizes large deviations. Consequently, extreme reserve prices are influenced by both the variance and the tail of the uncertainty distribution.

\begin{remark}
Prices in \eqref{LDT-CC specific prices} are derived from the resource perspective and capture the marginal cost of each generator. We can also derive the prices for energy and regular reserve from the perspective of the whole system, as:
\begin{subequations} \label{LDT-CC system-wise prices}
\begin{align}
    \pi &= \Bigg[ D- \hat{W} - \sum_{n \in \mathcal{N}} \frac{(C_{1,n} + \mu_n^\Plus + \delta_n^\Plus)}{2 C_{2,n}} \Bigg] \Bigg/ \sum_{n \in \mathcal{N}} \frac{1}{2 C_{2,n}} \label{LDT-CC system-wise energy prices}\\
    \rho &= \Bigg[ 1 - \sum_{n \in \mathcal{N}} \frac{( \mu_n^\Plus \hat{\sigma}_{n} + \delta_n^\Plus \hat{\sigma}_n)}{2 C_{2,n} \sigma_{\Omega}^2} \Bigg] \Bigg/ \sum_{n \in \mathcal{N}} \frac{1}{2 C_{2,n} \sigma_{\Omega}^2} \label{LDT-CC system-wise regular reserve prices}
\end{align}
\end{subequations} 
The expressions in \eqref{LDT-CC system-wise energy prices} and \eqref{LDT-CC system-wise regular reserve prices} are obtained by eliminating decision variables $p_n$ and $\alpha_n$ from \eqref{LDT-CC-energy-price} and \eqref{LDT-CC-regular-price}, respectively. We begin by summing over $n$ on both sides of \eqref{LDT-CC-energy-price} and \eqref{LDT-CC-regular-price}, and then substituting $\sum_{n \in \mathcal{N}} p_n$ with $D - \hat{W}$ and substituting $\sum_{n \in \mathcal{N}} \alpha_n$ with 1, respectively. Therefore, in a fully-competitive market, i.e., when the two conditions in Theorem~\ref{theo1} hold, the expressions in \eqref{LDT-CC specific prices} and \eqref{LDT-CC system-wise prices} should be equivalent.
\end{remark}

\subsection{Market-Clearing Based on LDT-WCC-ED}
\label{Subsection:LDT-WCC-pricing}
Since the LDT-WCC-ED model in \eqref{LDT-WCC-Model} includes the same constraints \eqref{LDT-CC-demand}-\eqref{LDT-CC-reserve_ad} as the LDT-CC-ED model in \eqref{LDT-CC}, the definitions of energy, regular reserve, and extreme reserve prices here are consistent with those in Section~\ref{Subsection:LDT-CC-pricing}. Therefore, in this subsection, we can perform a similar analysis based on LDT-WCC-ED, including deriving price formation and establishing market equilibrium.

We first derive the prices based on the LDT-WCC-ED model in \eqref{LDT-WCC-Model}, as stated in the following proposition:
\begin{proposition} \label{prop2}
   (Pricing based on LDT-WCC-ED) The optimal values of $\pi$, $\rho$, and $\chi$ based on \eqref{LDT-WCC-Model} are given by:
    \begin{subequations} \label{WCC-pricing}
    \begin{align} 
        \pi& = \frac{\partial C_n(\cdot)}{\partial p_n} + \left.\nu_{n}^{\Plus} \frac{\partial f(\cdot)}{\partial \Tilde{\mu}_n} \frac{\partial \Tilde{\mu}_n}{\partial p_n} \right|_{\text{I, II}} \\
        \rho &= \frac{\partial C_n(\cdot)}{\partial \alpha_n} + \left.\nu_{n}^{\Plus} \frac{\partial f(\cdot)}{\partial \Tilde{\mu}_n} \frac{\partial \Tilde{\mu}_n}{\partial \alpha_n}\right|_{\text{I, II}}  + \left.\nu_{n}^{\Plus} \frac{\partial f(\cdot)}{\partial \Tilde{\sigma}_n} \frac{\partial \Tilde{\sigma}_n}{\partial \alpha_n}\right|_{\text{I, II}} \\
        \chi &= \frac{\partial C_n(\cdot)}{\partial \beta_n} + \left.\nu_{n}^{\Plus} \frac{\partial f(\cdot)}{\partial \Tilde{\mu}_n} \frac{\partial \Tilde{\mu}_n}{\partial \beta_n}\right|_{\text{I, II}} + \left.\nu_{n}^{\Plus} \frac{\partial f(\cdot)}{\partial \Tilde{\sigma}_n} \frac{\partial \Tilde{\sigma}_n}{\partial \beta_n}\right|_{\text{I, II}}
    \end{align}
    \end{subequations}   
where the expressions for $\Tilde{\mu}n$ and $\Tilde{\sigma}n$ are given in \eqref{LDT-WCC-Mean} and \eqref{LDT-WCC-Variance}, and $f$ is defined in \eqref{unit WCC}. The derivatives of $f$ with respect to $\Tilde{\mu}_n$ and $\Tilde{\sigma}_n$ are provided in \eqref{df_dmu} and \eqref{df_dtheta}. The subscript \text{I,II} indicates that the expression consists of two parts, corresponding to the cases when $\boldsymbol{\Omega} \leq \Omega_{\epsilon}^{\Plus} $ and $\boldsymbol{\Omega} > \Omega_{\epsilon}^{\Plus} $, respectively.
\end{proposition}
\begin{proof}
   Similar to the proof of Proposition~\ref{prop1}, the stationary conditions of \eqref{LDT-WCC-Model} can be derived as follows:
    \begin{subequations} \label{ref1}
    \begin{align}
        & \frac{\partial C_n(\cdot)}{\partial p_n} + \left. \nu_{n}^{\Plus} \frac{\partial f(\cdot)}{\partial \Tilde{\mu}_n} \frac{\partial \Tilde{\mu}_n}{\partial p_n}\right|_{\text{I, II}} - \pi = 0 \label{ref11}\\
        &\frac{\partial C_n(\cdot)}{\partial \alpha_n}  -\rho \Plus\left. \nu_{n}^{\Plus} \Bigg( \frac{\partial f(\cdot)}{\partial \Tilde{\mu}_n} \frac{\partial \Tilde{\mu}_n}{\partial \beta_n} \Plus \frac{\partial f(\cdot)}{\partial \Tilde{\sigma}_n} \frac{\partial \Tilde{\sigma}_n}{\partial \alpha_n} \Bigg)\right|_{\text{I, II}} = 0 \label{ref12}\\
        & \frac{\partial C_n(\cdot)}{\partial \beta_n} - \chi \Plus \left.\nu_{n}^{\Plus} \Bigg( \frac{\partial f(\cdot)}{\partial \Tilde{\mu}_n} \frac{\partial \Tilde{\mu}_n}{\partial \beta_n}\Plus\frac{\partial f(\cdot)}{\partial \Tilde{\sigma}_n} \frac{\partial \Tilde{\sigma}_n}{\partial \beta_n} \Bigg)\right|_{\text{I, II}} = 0 \label{ref13},
    \end{align}
    \end{subequations}
    Thus, $\pi$, $\rho$ and $\chi$ can be expressed directly from \eqref{ref1}
\end{proof}
Next, we analyze how the prices derived from LDT-WCC-ED lead to a competitive equilibrium problem:
\begin{theorem} \label{theo2}
   (Market equilibrium based on LDT-WCC-ED) Let $\{p_n^*, \alpha_n^*, \beta_n^*\}$ denote the optimal solution to \eqref{LDT-WCC-Model} and let $\{ \pi^*, \rho^*, \chi^* \}$ be the  corresponding dual variables. Then $\{ \{ p_n^*, \alpha_n^*, \beta_n^* ~\forall n\}, \pi^*, \rho^*, \chi^* \}$ constitutes a market equilibrium, i.e.:
   \begin{itemize}
        \item The marker clears at $\sum p_n^* = D - \hat{W}$, $\sum \alpha_n^* =1$, and $\sum \beta_n^* =1$.
        \item Each producer maximizes its profit under the payment $\Gamma_n = \pi^* p_n^* + \rho^* \alpha_n^* + \chi^* \beta_n^*.$
    \end{itemize}
\end{theorem}

\begin{proof}
    The first condition can be easily proven using similar steps as in the proof of Theorem~\ref{theo1}. Next, we formulate the profit maximization problem for generator $n$ as follows:
    \begin{subequations} \label{WCC-KKT-firm}
        \begin{align}
        \max_{p_n,\alpha_n,\beta_n} ~&  - C_n(p_n,\alpha_n,\beta_n) + l_n p_n + b_n \alpha_n + w_n \beta_n \\
        \text{s.t.} ~ & ~ p_n, \alpha_n, \beta_n  \geq 0 \\
        (\hat \nu_n^\Plus): & \int_{\Minus\infty}^{\Omega_{\epsilon}^\Plus } \int_{0}^{\infty} y_n^\Plus \mathcal{P}(y_n^\Plus|\boldsymbol{\Omega}) \mathcal{P}(\boldsymbol{\Omega}) dy_n^\Plus d\boldsymbol{\Omega} \nonumber\\
        & + \int_{\Omega_{\epsilon}^\Plus }^{\infty} \int_{0}^{\infty} y_n^\Plus \mathcal{P}(y_n^\Plus|\boldsymbol{\Omega}) \mathcal{P}(\boldsymbol{\Omega}) dy_n^\Plus d\boldsymbol{\Omega}\leq \epsilon_n
    \end{align}
    \end{subequations}
    Based on the stationary conditions for \eqref{WCC-KKT-firm}, the prices $(l_n,b_n,w_n)$ can be derived as follows:
    \begin{subequations} \label{LDT-WCC prices individual}
    \begin{align} 
        l_n~& = \frac{\partial C_n(\cdot)}{\partial p_n} + \left. \hat\nu_{n}^{\Plus} \frac{\partial f(\cdot)}{\partial \Tilde{\mu}_n} \frac{\partial \Tilde{\mu}_n}{\partial p_n}\right|_{\text{I, II}}\\
        b_n ~&= \frac{\partial C_n(\cdot)}{\partial \alpha_n} +\left. \hat \nu_{n}^{\Plus} \Bigg( \frac{\partial f(\cdot)}{\partial \Tilde{\mu}_n} \frac{\partial \Tilde{\mu}_n}{\partial \alpha_n} + \frac{\partial f(\cdot)}{\partial \Tilde{\sigma}_n} \frac{\partial \Tilde{\sigma}_n}{\partial \alpha_n} \Bigg)\right|_{\text{I, II}}\\
        w_n ~&= \frac{\partial C_n(\cdot)}{\partial \beta_n} +\left. \hat\nu_{n}^{\Plus} \Bigg(\frac{\partial f(\cdot)}{\partial \Tilde{\mu}_n} \frac{\partial \Tilde{\mu}_n}{\partial \beta_n} + \frac{\partial f(\cdot)}{\partial \Tilde{\sigma}_n} \frac{\partial \Tilde{\sigma}_n}{\partial \beta_n} \Bigg)\right|_{\text{I, II}}
    \end{align}
    \end{subequations}
    By comparing the prices in \eqref{WCC-pricing} with those in \eqref{LDT-WCC prices individual}, we find that $\pi = l_n$, $\rho = b_n$, $\chi =w_n, \forall n \in \mathcal{N}$. Thus, the market clears at prices $(\pi, \rho, \chi)$ that enable each producer to maximize their profit.
\end{proof}
Due to the presence of the truncated Gaussian distribution in the LDT-WCC, the prices $\pi$, $\rho$, and $\chi$ depend on $\Phi(z_n)$ and $\phi(z_n)$, where $z_{n} = -\Tilde{\mu}_n / \Tilde{\sigma}_n$. As a result, the price expressions in \eqref{WCC-pricing} cannot be simplified in the same manner as in \eqref{LDT-CC specific prices}. This limitation prevents a clear interpretation of the price components and a theoretical comparison between the prices derived from LDT-CC-ED and LDT-WCC-ED. In Section~\ref{Sec:CaseStudy}, we will compare these prices through numerical experiments.

\subsection{Market Properties of LDT-CC-ED and LDT-WCC-ED}
\label{Subsection:market properties}
Using the cost function $C_n(p_n, \alpha_n, \beta_n) = C_{1,n} (p_n + \boldsymbol{\Omega}\alpha_n) + C_{2,n} (p_n + \boldsymbol{\Omega} \alpha_n)^2 +  C^{\beta}_{n} \beta_n$, we first analyze market properties, including cost recovery and revenue adequacy, based on LDT-CC-ED.

\subsubsection{Cost Recovery} Cost recovery refers to the ability of producers to recover their operational cost from the market outcomes. It is formalized as, $\Pi_n \geq 0$, $\forall n \in \mathcal{N}$, where $\Pi_n = \pi p_n + \rho \alpha_n + \chi \beta_n - C_{1,n} p_n - C_{2,n}(p_n^2 + \sigma_{\Omega}^2 \alpha_n^2) - C^{\beta}_{n} \beta_n$.
Theorem \ref{theo1} guarantees full cost recovery for each producer, i.e., $\Pi_n^* = 0$, under the competitive equilibrium. Since each producer problem in \eqref{LDT-KKT-firm} is convex, we can apply the strong duality theorem to calculate the optimal market outcomes as discussed in \cite{dvorkin2019chance}.
\subsubsection{Revenue Adequacy}
Revenue adequacy refers to the market ability to ensure that the total payments received from consumers are sufficient to cover the total payments to producers. The proposed market design follows the same principles as those in \cite{dvorkin2019chance}, and results in a revenue inadequate market. The total market revenue deficit is given by:
\begin{align}
    \Delta^*=-\min \big[ 0, \sum_{n} \Gamma_n + \pi^* \hat{W} - \pi^* D \big] \label{RA}
\end{align}
From Theorem \ref{theo1} we define $\Gamma_n = \pi^* p_n + \rho^* \alpha_n^* + \chi^* \beta_n^*$, and establish that  $\sum_{n} \alpha_n^* = 1$, $\sum_{n} \beta_n^* = 1$, and $\sum_{n} p_n^* = (D - \hat{W})$. Then, \eqref{RA} can be expressed as:
\begin{align}
    \Delta^*=-\min \big[ 0, - \rho^* - \chi^* \big]
\end{align}
Since $\rho^*$ and $\chi^* \geq 0 $, the market revenue in \eqref{RA} results in a deficit $\Delta^* \geq 0$. Hence, the market design requires further allocation among customers \cite{dvorkin2019chance}.

Similarly, due to the convexity of the producer's problem in \eqref{WCC-KKT-firm}, using strong duality, Theorem \ref{theo2} ensures full cost recovery by each producer under a competitive equilibrium. Also, the proposed market design follows the same principles as those in \eqref{RA}, resulting in a revenue inadequate market with a total deficit $\Delta^* = \Minus \min[0, \Minus \rho^* \Minus \chi^*].$

\section{Network-Constrained Extension} \label{Sec:NetworkExtension}

This section studies the extension of the proposed LDT-CC-ED model by including network information and power flow constraints. In Section~\ref{Subsection:System-Wide-Reserve}, we introduce network constraints into the LDT-CC-ED model with a system-wide reserve requirement and demonstrate that (i) energy prices take the form of locational marginal prices (LMPs) and (ii) the results of Propositions \ref{prop1} and Theorem \ref{theo1} remain valid. The additional network constraints lead to the LDT-CC optimal power flow (LDT-CC-OPF) model that captures the locational impacts of wind power uncertainties on the scheduling of both regular and extreme reserves. To enhance real-time reserve deliverability, we further extend the LDT-CC-OPF model to incorporate location-specific reserve requirements in Section~\ref{Subsection:Location-Specific-Reserve}. This extension provides more detailed information of uncertainty to the reserve scheduling model, leading to more efficient reserve allocation.

\subsection{System-Wide Reserve Requirements}
\label{Subsection:System-Wide-Reserve}
We first consider reserve scheduling based on $\Omega$, i.e., the total uncertainty in the system. The LDT-CC-OPF model with system-wide reserve requirements is formulated as:
\begin{subequations} \label{LDTED_Network}
\begin{align}
\min_{p,\alpha,\beta, \Omega^*_\Plus, \lambda^{\Plus *}} ~& \mathbb{E}_{\Omega} \big[\sum_{n \in \mathcal{N}} C_n(p_n,\alpha_n, \beta_n) \big]\\
\text{s.t.}~& ~\text{\eqref{LDT-CC-T0}}-\text{\eqref{LDT-CC-R3}},~\text{\eqref{LDT-CC-reserve_ad}}-\text{\eqref{LDT-CC-reserve_op}} \nonumber \\
(\pi_i): ~ &\sum_{n \in \mathcal{N}_{i}} p_{n} + \sum_{j \in \mathcal{L}^{+}_{i}} f_{jk} \nonumber \\
&  - \sum_{j \in \mathcal{L}^{-}_{i}} f_{jk} = d_i - \hat W_i \quad \forall i \in \mathcal{I} \label{LDTED_net-demand}\\
(\eta_{jk}^{\Minus}):~& - f^{\max}_{jk} \leq f_{jk} \quad \forall (j,k) \in \mathcal{L} \label{LDTED_net_flow_max}\\
(\eta_{jk}^{\Plus}):~& f_{jk} \leq f^{\max}_{jk}  \quad \forall (j,k) \in \mathcal{L} \label{LDTED_net_flow_min} \\
(\eta_{jk}^{0}):~&  B_{jk}(\theta_j - \theta_k) = f_{jk}  \quad \forall (j,k) \in \mathcal{L}, \label{LDTED_network_angles}
\end{align}
\end{subequations}
where set $\mathcal{I}$ collects all the nodes, $\mathcal{L}$ collects all the lines, $\mathcal{L}^{\Plus}_i$ and $\mathcal{L}^{\Minus}_i$ are the sets of lines that are connected from the node $i$ and to the node $i$, respectively. Eq. \eqref{LDTED_net-demand} is the power balance constraint that replaces \eqref{LDT-CC-demand}. Eq. \eqref{LDTED_net_flow_max} and \eqref{LDTED_net_flow_min} define the line power flow limits, and\eqref{LDTED_network_angles} establishes the relationship between node angles and line power flow.

\begin{theorem}
    Consider the model in \eqref{LDTED_Network}. Then (i) energy prices $\pi$ from Proposition \ref{prop1} (energy price) become LMPs $\pi_i$ and (ii) the results of Propositions \ref{prop1} (reserve prices) and Theorem \ref{theo1} (market equilibrium) remain valid.
\end{theorem}

\begin{proof}
The KKT stationary conditions of \eqref{LDTED_Network} are:
\begin{subequations} \label{LDT-KKT-net}
\begin{align}
    \frac{\partial \mathcal{L}}{\partial p_n}:~& \frac{\partial C_n(\cdot)}{\partial p_n} + \mu_{n}^{\Plus} + \delta_{n}^{\Plus} - \pi_i = 0 \label{R11}\\
    \frac{\partial \mathcal{L}}{\partial \alpha_n}:~&\frac{\partial C_n(\cdot)}{\partial \alpha_n} -\mu_{n}^{\Plus} \hat{\sigma}_{n} + \delta_{n}^{\Plus} \hat{\sigma}_n -\rho = 0 \label{R12}\\
    \frac{\partial \mathcal{L}}{\partial \beta_{n}}:~&\frac{\partial C_n(\cdot)}{\partial \beta_{n}} -\mu_{n}^{\Plus} (\Omega^*_\Plus - \hat{\sigma}_{n}) + \xi^\Plus_{n} \lambda^{\Plus *}_n - \chi = 0 \label{R13}\\
    \frac{\partial \mathcal{L}}{\partial \theta_j}:~& B_{jk} \eta_{jk}^{0} = 0\\
    \frac{\partial \mathcal{L}}{\partial f_{jk}}:~& B_{jk}(\pi_j - \pi_k) + \eta_{jk}^{\Plus} - \eta_{jk}^{\Minus} - \eta_{jk}^{0} = 0,
\end{align}
\end{subequations}
The KKT stationary conditions associated with energy and reserve balances in \eqref{R11}-\eqref{R13} are analogous to those in \eqref{LDT-KKT}. Thus, $\pi_i$ can be obtained directly from \eqref{R11}, while $\rho$ and $\chi$ can be derived from \eqref{R12} and \eqref{R13}, respectively. Consequently, Proposition \ref{prop1} remains valid for \eqref{LDTED_Network}.
\end{proof}
\subsection{Location-Specific Reserve Requirements}
\label{Subsection:Location-Specific-Reserve}
Instead of setting a system-wide reserve requirement based on the total uncertainty $\Omega$, we can increase the granularity by scheduling reserves according to $\Omega_i$, where $ \boldsymbol{\Omega}_{i} = \sum_{n' \in \mathcal{W}_{i}} \boldsymbol{\omega}_{n'}, \forall i \in \mathcal{I}$ is the aggregated wind power uncertainty at the nodal level. The LDT-CC-OPF model with location-specific reserve requirements is then given by:
\begin{subequations} \label{LDTED_Network_locational}
\begin{align}
\min_{p, A, B,\Omega^*_\Plus,\lambda^{\Plus *}}&\mathbb{E}_{\Omega} \sum_{n \in \mathcal{N}} C_n(p_n, A_n, B_n)\\
 \text{s.t.}&~ \eqref{LDTED_net-demand}-\eqref{LDTED_network_angles} \nonumber \\
&~\forall n ~\Big\{A_{n}, B_{n}, p_n \geq 0, \lambda^{\Plus *}_n > 0  \label{LDT_locational_net-T0}\\
(\delta_{n}^{\Plus}):~& p_n - p_{n}^{\max} + G(A_n) \leq 0 \label{LDT_locational_net-T1}\\
(\mu_{n}^{\Plus}):~& p_n + G(A_n \Minus B_n) + B_n \Omega^*_\Plus = p_n^{\max}\label{LDT_locational_net-R1}\\
(\xi_n^\Plus):~&\Sigma_{\Omega}^{\Minus1} \Omega^*_\Plus - B_n \lambda^{\Plus *}_n = 0 \Big\} \label{LDTED_network-_R3}\\
(\kappa^\Plus):~& \Minus\Sigma_{\Omega}^{\Minus\frac{1}{2}} \Omega^*_\Plus - \Phi^{\Minus 1}(\epsilon^{\text{ext}}) \leq 0 \label{LDTED_network-OO} \\
(\rho_i):~&\sum_{n\in \mathcal{N}} A_{ni} = 1 \qquad \forall i \In \mathcal{I}\\
(\chi_i):~&\sum_{n\in \mathcal{N}} B_{ni} = 1 \qquad \forall i \In \mathcal{I},
\end{align}
\end{subequations}
Model \eqref{LDTED_Network_locational} shares the same set of power flow constraints as in \eqref{LDT-CC}, but it distinguishes the reserve contribution of each generator to hedge uncertainty at different nodes. Matrices $A = [A_{(n,i)}]_{n \in \mathcal{N}, i \in \mathcal{I}}$ and $B = [B_{(n,i)}]_{n \in \mathcal{N}, i \in \mathcal{I}}$ represent the regular and extreme reserve participation factors, respectively, for each generator $n$ in controlling wind deviations at each node $i$. For clarity, we use $A_n = A_{(n,\cdot)}$ and $B_n = B_{(n,\cdot)}$ to denote the $n$-th columns of $A$ and $B$, respectively, and $A_{(\cdot,i)}$ and $B_{(\cdot,i)}$ to represent the $i$-th row. 
In addition, the statistical characteristics of uncertainty at each node need to be quantified separately, where $\Sigma_{\Omega}$ is the covariance matrix with diagonal $\sigma_{(\Omega,i)}^2$, $\Omega^*_\Plus$ is a vector of dominating points in the set $\boldsymbol{\Omega}$ for each node $i$. For simplicity, we use $G(X_n) = \Phi^{\Minus1}(1\Minus\epsilon_n) \sqrt{X_{n}^{\top} \Sigma_{\Omega} X_{n}}$. This model will be used in Section~\ref{Subsection:ISONE_case} to conduct numerical experiments on the 8-node ISO New England system.

\section{Case Study}
\label{Sec:CaseStudy}

In this section, we test the performance of the proposed market-clearing models through numerical experiments. We begin with a single-node illustrative example in Section~\ref{Subsection:illustrative_case} to compare energy dispatch, reserve allocations, and the corresponding prices across the CC-ED, LDT-CC-ED and LDT-WCC-ED formulations. Then, we expand this comparison to the 8-node ISO New England system in Section~\ref{Subsection:ISONE_case}, demonstrating the scalability of the formulations and incorporating locational variability into the analysis. All simulations were carried out in Python using the Gurobi solver \cite{Gurobi2022}. Since all our models are convex, all problems were solved with a duality gap $< 0.01\%$. The code for the case study is available at \cite{Tapia2024}. 

\subsection{Illustrative Example} \label{Subsection:illustrative_case}
We first consider an illustrative single-node system with three controllable generators and one wind farm, as shown in Fig.~\ref{fig:illustrative}). The total demand is 270 MW, while the cost of energy not served is 9000 $\$$/MWh. The quadratic component of the cost is $C_2^{\top} = [0.01,0.05,0.025]$ $\$$/MWh$^2$, while the linear component is $C_1^{\top} = [10,35, 50]$ $\$$/MWh. The extreme reserve cost is $C_{\beta}^{\top}= [700,300,600]$ $\$$. The maximum capacities of generators are $(p^{\max})^{\top} = [75, 160, 120]$ MW. Finally, the wind power forecast $\hat{W}$ is 150 MW, and the forecast error is zero-mean with $\sigma_{\Omega}=50$ MW. The risk tolerance is set as $\epsilon_n = 0.05$ and $\epsilon^{\text{ext}}_n = 5\times10^{-5}, \forall n$.
\begin{figure}[!t]
    \centering
    \includegraphics[width=0.27\textwidth]{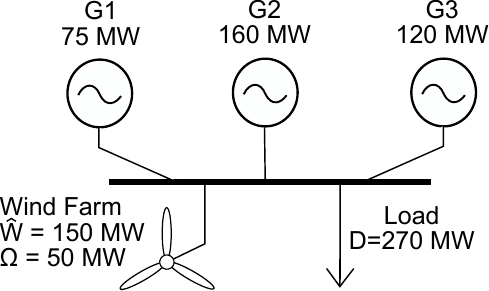}
    \caption{Illustrative single-node system.}
    \label{fig:illustrative}
\end{figure}

\begin{table} [!t]
\centering
\caption{Illustrative Example Optimal Dispatch Results} 
\label{Tab:StudyCases1Primal} 
\begin{tabular}
{ p{1.3cm}||p{1.4cm}|p{1.2cm}|p{1.2cm}|p{1.2cm}}
 \hline
 \multirow{2}{*}{Model} &\multicolumn{4}{|c}{ Energy ($p$) \& Reserves ($\alpha, \beta$) dispatch }\\
  & Variables & G1  & G2 & G3 \\
 \hline
 \multirow{2}{*}{CC} & $p^*$ [MW] & 75  & 45  & 0 \\
 & $\alpha^*$ [$\%$]  & 0   & 33   & 67 \\
  \hline
\multirow{3}{*}{LDT-WCC} & $p^*$ [MW]  & 75 & 45 & 0 \\
 & $\alpha^*$ [$\%$]  & 0  & 44  & 56 \\
 & $\beta^*$ [$\%$]   & 0  & 100  & 0 \\
\hline
\multirow{3}{*}{LDT-CC} & $p^*$ [MW] & 75 & 45  & 0 \\
 & $\alpha^*$ [$\%$]  & 0  & 0 & 100 \\
 & $\beta^*$ [$\%$]   & 0  & 75  & 25 \\
 \hline
\end{tabular}
\end{table}

Table \ref{Tab:StudyCases1Primal} compares the optimal dispatch for each model in terms of the energy, regular and extreme reserve allocations. All models assign the same energy dispatch but differ in reserve allocations, which are driven by different model conservatism. Recall that the LDT-WCC and LDT-CC models incur an additional cost at the scheduling stage due to the provision of extreme reserves. We note that LDT-WCC allocates slightly more regular reserve to G2 than the CC model, with G2 providing $44\%$ of the regular reserve and $100\%$ of the extreme reserve. In constrast, LDT-CC assigns only $75\%$ of the extreme reserve to G2, while G3, with a higher production cost, complements this. The less restrictive requirements in LDT-WCC allows for minimizing the cost of reserve provision by allowing G2 to provide more than $40\%$ of the regular and the whole extreme reserve requirement. In contrast, LDT-CC is more restrictive in its reserve constraints, causing both G2 and G3 to provide extreme reserve.

\begin{table} [!t]
\centering
\caption{Optimal Prices and Total System Cost}
\label{Tab:StudyCases1Dual_p} 
\begin{tabular}{l||l|l|l|l}
\hline
\multirow{2}{*}{Model}&\multicolumn{4}{|c}{Energy ($\pi$) \& Reserves ($\rho, \chi$) prices}\\
 & $\pi^*$ [$\$$/MW]& $\rho^*$ [$\$$/$\%$] & $\chi^*$ [$\$$/$\%$] & T. Cost [$\$$]\\
 \hline
 CC   & 39.20  & 83.33 &  - & 2524.17\\
 LDT-WCC & 39.50 & 109.25 & 300.00 & 2826.16\\
 LDT-CC  & 41.47 & 125.74 & 601.37 & 2919.15\\
 \hline
\end{tabular}
\end{table} 

Table \ref{Tab:StudyCases1Dual_p} presents energy and reserve prices along with the total system cost. As expected, the total system cost at the scheduling stage increases with model conservatism. Specifically, the cost for the LDT-WCC and LDT-CC models are $12.0\%$ and $15.6\%$ higher than the CC benchmark, respectively. The differences between the energy dispatch and reserve allocation in Table \ref{Tab:StudyCases1Primal} are reflected in the prices in Table \ref{Tab:StudyCases1Dual_p}. Compared to the CC model, the energy price increases slightly for LDT-WCC and LDT-CC by $0.7\%$ and $5.6\%$, while the regular reserve price increases by $31.1\%$ and $50.8\%$, respectively.

\subsection{ISO New England Case Study} \label{Subsection:ISONE_case}
We further extend the numerical experiment to a network-constrained system. Fig.~\ref{fig:enter-label} shows the 8-zone ISO New England system used in this study with the data from \cite{Krishnamurthy2016ISONEData}. The wind power forecast is $\hat{W} = 3600$ MW, and the forecast error is zero-mean with $\sigma_{\Omega}=1100$ MW. We set $\epsilon_n = 0.05$ and $\epsilon^{\text{ext}}_n = 5\times10^{-5}, \forall n$.
\begin{figure}[!t]
    \centering
    \includegraphics[width=0.27\textwidth, trim={140pt 85pt 0 72pt},clip]{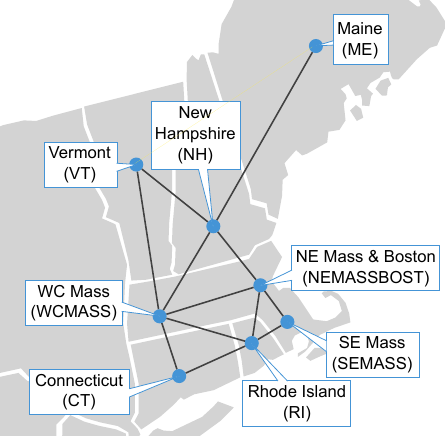}
    \caption{8-zones ISO New England system \cite{Krishnamurthy2016ISONEData}.}
    \label{fig:enter-label}
\end{figure}
Fig.~\ref{fig:extendedplots} summarizes the optimal dispatch for all three models in terms of energy, regular and extreme reserve allocated to generators in each zone. We observe that the energy dispatch remains consistent across all three models. The zones allocated to provide regular and extreme reserve are the same, but the allocations differ based on how the model addresses the burden of coping with extreme events. For instance, compared to CC, LDT-WCC increases the regular reserve allocation in ME, while LDT-CC increases it in RI. However, the allocation of extreme reserve remains unchanged. LDT-CC diversifies the reserve provision by assigning more than $75\%$ to three different zones, whereas LDT-WCC allocates this reserve only in ME, leveraging cheaper generators.
\begin{figure}[!t]
     \begin{subfigure}[b]{0.49\textwidth} 
         \centering
         \includegraphics[width=\textwidth, trim={0 6pt 0 0},clip]{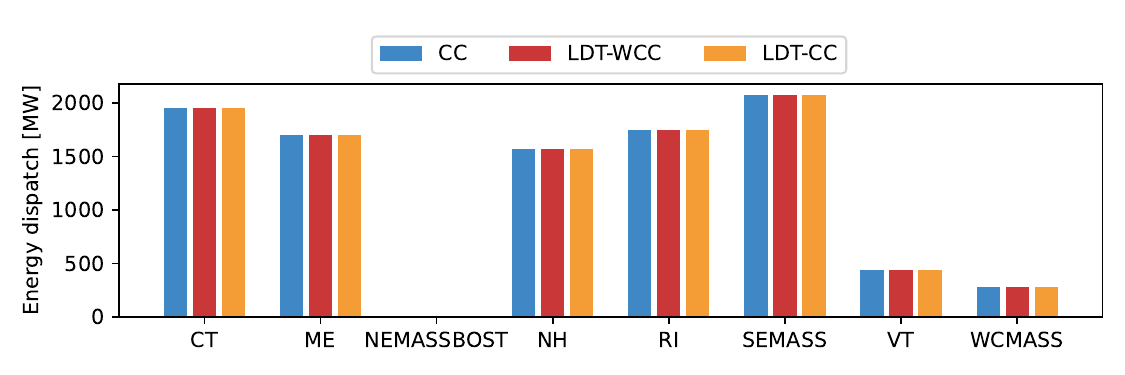}
         \caption{Energy dispatch}
         \smallskip
         \label{fig:F3a}
     \end{subfigure}
     \begin{subfigure}[b]{0.49\textwidth} 
         \centering
         \includegraphics[width=\textwidth, trim={0 6pt 0 40pt},clip]{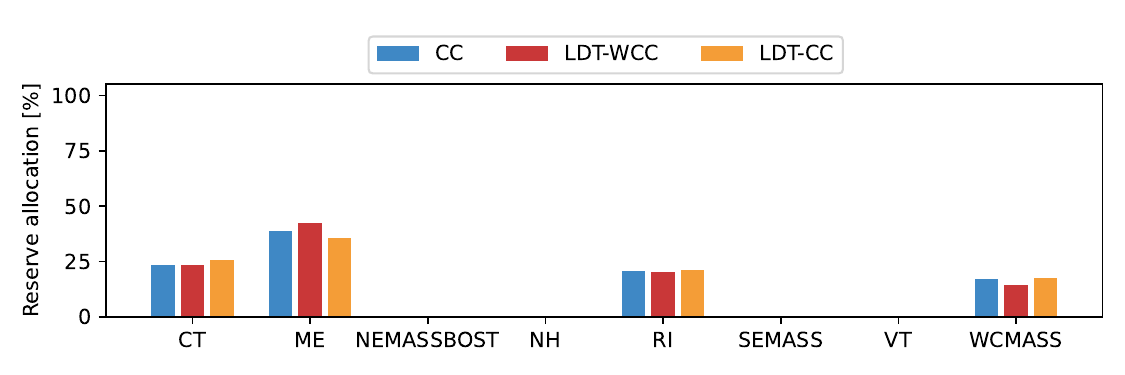}
         \caption{Regular reserve}
         \smallskip
         \label{fig:F3b}
     \end{subfigure}
     \begin{subfigure}[b]{0.49\textwidth}
         \centering
         \includegraphics[width=\textwidth, trim={0 6pt 0 40pt},clip]{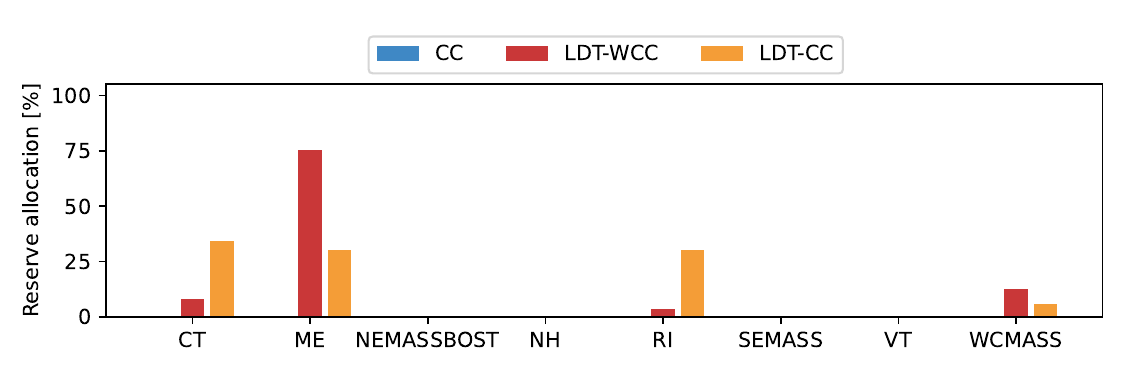}
         \caption{Extreme reserve}
         \smallskip
         \label{fig:F3c}
     \end{subfigure}
        \caption{Energy dispatch, regular and extreme reserve comparison.}
        \label{fig:extendedplots}
\end{figure}
Table \ref{Tab:StudyCases2Dual_p} compares the energy and reserve prices. All three formulations yield the same energy prices, matching the dispatch outcomes in Fig. \ref{fig:F3a}. LDT-WCC increases the regular reserve price by 11$\%$ relative to CC, while  LDT-CC raises the regular reserve price by 48$\%$. Compared to LDT-WCC, LDT-CC results in an 800$\%$ increase in the extreme reserve price. Table \ref{Tab:StudyCases1Profit2} summarizes the optimal revenues, total costs, and profits obtained for the outcomes detailed in Table \ref{Tab:StudyCases2Dual_p} and Fig. \ref{fig:extendedplots}. We observe that the difference between LDT-WCC and LDT-CC is reflected in the total cost in ME, but not in other zones. This difference also also results in a higher total profit for all zones, attributed to the increase in energy and reserve prices in the more conservative models. 
\begin{table} [!t]
\centering
\caption{Optimal Dual Results $\pi^*~[\$$/MW], $\rho^*$ $\&~\chi^*$ $[\$$/$\%$] }  \label{Tab:StudyCases2Dual_p} 
\begin{tabular}{ p{1.1cm}||p{1.6cm}||p{1.1cm}|p{1.3cm}|p{1.1cm}}
\hline
\multirow{2}{*} {Product} & \multirow{2}{*} {Price} & \multicolumn{3}{c}{Model} \\
 & &  CC & LDT-WCC & LDT-CC\\
\hline
\hline
\multirow{8}{*}{Energy}
& $\pi^*_{CT}$      & 58.23   & 58.26  &  58.33\\
& $\pi^*_{ME}$      & 135.96  & 136.01 &  136.05\\
& $\pi^*_{NEMASSB}$ & 162.62  & 162.66 &  162.68 \\
& $\pi^*_{NH}$      & 135.96  & 136.01 &  136.05 \\
& $\pi^*_{RI}$      & 33.01   & 33.08  &  33.20 \\
& $\pi^*_{SEMASS}$  & 259.89  & 259.92 &  259.85 \\
& $\pi^*_{VT}$      & 123.65  & 123.72 &  123.76 \\
& $\pi^*_{WCMASS}$  & 105.20  & 105.27 &  105.32 \\
\hline
Regular & \multirow{2}{*}{$\rho^*$}  & \multirow{2}{*}{713.93}  & \multirow{2}{*}{792.65} & \multirow{2}{*}{1056.58}\\
Reserve & & & & \\
\hline
Extreme & \multirow{2}{*}{$\chi^*$} & \multirow{2}{*}{-}  & \multirow{2}{*}{131.34} & \multirow{2}{*}{1182.41} \\
Reserve & & & & \\
\hline
\end{tabular}
\end{table}

\begin{table}[!t]
\centering
\caption{Optimal Daily Revenue, Cost, and Profit (in \$)} \label{Tab:StudyCases1Profit2}
\resizebox{0.85\linewidth}{!}{%
\begin{tabular}{ l||l|l|l|l}
 \hline
 \multirow{2}{*}{Zone} & \multirow{2}{*}{Metric} & \multicolumn{3}{c}{Model} \\
 & & CC & LDT-WCC & LDT-CC \\
 \hline
  \hline
 \multirow{3}{*}{CT} & Revenue & 113822 & 113914 & 114612 \\
 & Cost & 79641 & 79671 & 79837 \\
 & Profit & 34181 & 34243 & 34776 \\
 \hline
 \multirow{3}{*}{ME} & Revenue & 231180 & 231433 & 231788 \\
 & Cost & 32280 & 32414 & 32302 \\
 & Profit & 198900 & 199019 & 199485 \\
 \hline
 \multirow{3}{*}{NEMASSB} & Revenue & 0 & 0 & 0 \\
 & Cost & 0 & 0 & 0 \\
 & Profit & 0 & 0 & 0 \\
 \hline
 \multirow{3}{*}{NH} & Revenue & 212810 & 212909 & 212962 \\
 & Cost & 47055 & 47055 & 47055 \\
 & Profit & 165764 & 165854 & 165907 \\
 \hline
 \multirow{3}{*}{RI} & Revenue & 57849 & 57982 & 58614 \\
 & Cost & 41939 & 41960 & 42075 \\
 & Profit & 15910 & 16022 & 16539 \\
 \hline
 \multirow{3}{*}{SEMASS} & Revenue & 538071 & 538127 & 537990 \\
 & Cost & 229079 & 229080 & 229079 \\
 & Profit & 308991 & 309047 & 308910 \\
 \hline
 \multirow{3}{*}{VT} & Revenue & 53683 & 53710 & 53728 \\
 & Cost & 5291 & 5294 & 5294 \\
 & Profit & 48389 & 48416 & 48435 \\
 \hline
 \multirow{3}{*}{WCMASS} & Revenue & 30019 & 30047 & 30184 \\
 & Cost & 9084 & 9099 & 9105 \\
 & Profit & 20935 & 20948 & 21079 \\
 \hline
\end{tabular}}
\end{table}

To assess adaptability of the market outcomes under each formulation, we compare the cost performance across 3000 wind scenarios in Fig. \ref{fig:enter-label-4}. The blue bar represents the total cost of the scheduled operation, while the orange bar represents the average total cost of the 3000 out-sample scenarios. The red line indicates the standard deviation. We can observe that CC has the highest expected cost ($2.30$ million $\$$) with a standard deviation of $1.72$ million $\$$. This high cost is because CC is incomplete relative extreme deviations, which results in insufficient reserve procurement and, consequently, unserved energy. In comparison to CC, LDT-WCC and LDT-CC have expected costs that are $47\%$ and $26\%$ lower with standard deviation of $0.43$ and $0.83$ million $\$$, respectively. Thus, LDT-WCC and LDT-CC reduce the exposure to extreme event realizations more effectively than the CC benchmark.
\begin{figure}[!t]
    \centering
    \includegraphics[width=0.45\textwidth]{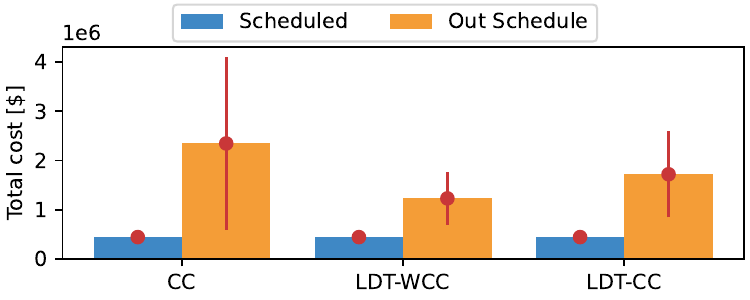}
    \caption{Expected and standard deviation cost performance for the 3000 scenarios in the ISO New England system.}
    \label{fig:enter-label-4}
\end{figure}
\section{Conclusion and Future Work} \label{Sec:Conclusion}

This paper proposes mathematical models that effectively account for the risk of extreme events in power system day-ahead scheduling and introduces a market design that co-optimizes the procurement and pricing of energy, regular reserves, and extreme reserves. First, we propose the LDT-CC-ED model for extreme reserve scheduling and then reformulate it into a single-level optimization problem that can be solved by commercial solvers. Additionally, recognizing that extreme reserve scheduling based on LDT-CC-ED may be overly conservative, leading to surging operational costs, we propose a more flexible LDT-WCC-ED model, offering system operators the chance to balance reliability and cost. The LDT-WCC-ED model can be efficiently solved using the proposed cutting-plane algorithm, ensuring its practicality in real systems. Finally, we derive the marginal prices of energy, regular reserves, and extreme reserves, demonstrating key market properties such as competitive equilibrium, cost recovery, and revenue adequacy, further supporting the applicability of the proposed market design.

Nevertheless, we must acknowledge that there remains a gap between the proposed market design and the current market operation. This is partly because most existing markets are based on deterministic models rather than stochastic optimization, such as the chance-constrained models. Additionally, to fully demonstrate the advantages of the proposed model through actual market performance, it would need to be implemented in the real system over an extended period. The benefits of scheduling extreme reserve using the proposed model only become evident when extreme events occur, and due to the rarity and unpredictability of such events, it is difficult to show significant cost savings by comparing operating costs with benchmark costs in a short timeframe.

Our future work on this research topic will focus on extending the proposed pricing theory to multi-period and security-constrained market-clearing tools, as well as analyzing multi-period cost recovery and revenue adequacy properties. This could involve, for example, developing  mixed-integer second-order conic or copositive programs that would require additional approximation methods to solve. Parallel work will include correlation analysis between different locations of uncertainty sources and the development of market designs  ensuring a competitive equilibrium.

\section*{Acknowledgement}
This publication is based upon work supported by the King Abdullah University of Science and Technology under Award \#ORFS-2022-CRG11-5021, US Department of Energy Advanced Research Projects Agency–Energy under Grant\# DEAR0001300, and the National Science Foundation under Award \#OISE 2330450.

\bibliographystyle{ieeetr} 
\bibliography{bio}

\end{document}